\documentclass[journal,onecolumn]{IEEEtran}
\usepackage{amsmath,amssymb,epsfig,psfrag,cite,subfigure}
\usepackage{setspace}
\include{macros}
\usepackage{graphicx}
\usepackage{bbm}

\usepackage[linesnumbered,ruled,vlined]{algorithm2e}

\usepackage{algpseudocode}

\usepackage{subfigure}
\usepackage{color}

\usepackage{tikz}
\usepackage{pict2e}
\usepackage{amsmath,amsfonts,amsthm} % Math packages
\usepackage{commath}
\usepackage{pgfplots}
\usepackage{bbm}

\newtheorem{theorem}{Theorem}
\newtheorem{lemma}{Lemma}

\newtheorem{corollary}{Corollary}
\newtheorem{definition}{Definition}
\newtheorem{remark}{Remark}

\newtheorem{property}{Property}

\newcommand{\bs}[1]{\boldsymbol{#1}}

\def \cX {\mathcal{X}}

\def \bS {\boldsymbol{S}}

\def \cB {\mathcal{B}}

\def \cF {\mathcal{F}}
\def \cV {\mathcal{V}}
\def \cS {\mathcal{S}}

\def \cT {\mathcal{T}}
\def \vmin {v_{\min}}
\def \indic {\mathbbm{1}}
\DeclareMathOperator{\?}{?}

\DeclareMathOperator{\cost}{cost}
\DeclareMathOperator{\parent}{parent}
\DeclareMathOperator{\parentone}{parent_{\cB_1}}
\DeclareMathOperator{\temp}{temp}

\algdef{SE}[DOWHILE]{Do}{doWhile}{\algorithmicdo}[1]{\algorithmicwhile\ #1}

\begin{document}

	\title{Optimal Policies for Age and Distortion in a Discrete-Time Model}
		\author{\IEEEauthorblockN{Yunus \.Inan, \IEEEmembership{Student Member, IEEE}, Reka Inovan, \IEEEmembership{Student Member, IEEE}, Emre Telatar, \IEEEmembership{Fellow, IEEE}}\thanks{The authors are with  \'{E}cole Polytechnique F\'{e}d\'{e}rale de Lausanne (EPFL), 1015 Lausanne, Switzerland. Emails: \{yunus.inan, reka.inovan, emre.telatar\}@epfl.ch.\newline A short version of this work is presented at IEEE ITW 2021 \cite{InanITW}.}}
	
	\maketitle 
	\begin{abstract}
%		We propose a simple model to study the tradeoff between timeliness and distortion, where different pieces of data have a different cost of not being sent.  We pose the question of finding the optimal tradeoff as a policy design problem amenable to dynamic programming methods.  We study the structural properties of optimal transmission policies, give an algorithmic procedure to find the optimal tradeoff, and numerically evaluate some instances
We study a discrete-time model where each packet has a cost of not being sent --- this cost might depend on the packet content. We study the tradeoff between the age and the cost where the sender is confined to packet-based strategies. The optimal tradeoff is found by an appropriate formulation of the problem as a Markov Decision Process (MDP). We show that the optimal tradeoff can be attained with finite-memory policies and we devise an efficient policy iteration algorithm to find these optimal policies. We further study a related problem where the transmitted packets are subject to erasures. We show that the optimal policies for our problem are also optimal for this new setup. Allowing coding across packets significantly extends the packet-based strategies. We show that when the packet payloads are small, the performance can be improved by coding.

	\end{abstract}
\vspace*{-2pt}
	\section{Introduction}
	Timeliness of information is a crucial aspect of communications. Stale data may have highly undesirable effects; think, for example, of sensor output for self-driving vehicles, position of an airplane, coolant temperature in a power plant, etc.  This aspect of data is nicely captured by the recently studied notion of Age-of-Information (AoI), by shifting the focus from delay to freshness.  At the same time, not all data is equally important.  If, in an attempt to reduce staleness our system drops important pieces of data, the remedy may be worse than the disease.  In this work, we study a simple setup where the freshness and importance aspects may be treated together.
	
	The loss, or misrepresentation of data and assigning higher cost to more important data is well-captured by the tools of rate-distortion theory.  As said above, the question of freshness has been an object of study in the AoI literature initiated by Kaul et al. \cite{Kaul}. Since the introduction of AoI, there has been various uses of this metric in many applications. In this work, we quantify the notion of importance by using a distortion metric. We analyze a discrete-time model which allows us to analyze the tradeoff between timeliness as measured by AoI and the distortion of the data. The tradeoff can be studied by casting it as a problem of finding the optimal policy of a Markov Decision Process (MDP) which identifies the packets to be dropped. We show that the optimal policy for this MDP can be achieved by a system with finite memory and we also provide an explicit algorithm to compute this policy, which also turns out to be an optimal policy of a problem where the sender transmits packets over an erasure channel with feedback. Lastly, for packets with small payload, we show that one can improve the performance with variable-to-fixed length coding techniques, such as Tunstall coding. 
	
	\section{Related Work}
	
	Freshness of data is recently recognized as a semantic aspect of communication. Initial work by Kaul et al. introduced the AoI as a metric to quantify this aspect \cite{Kaul,Kaul2,Kaul3}. Following Kaul et al., there have been many studies adopting AoI as a freshness metric. The first strand involved calculation of AoI in simple schemes, e.g. M/M/1 queues \cite{Kaul}. Subsequent extensions involve more general queues \cite{8406909,GG11,generalFormula,MMinfty}, multiple source streams \cite{multisource_Yates,multiplesource,MG11_multistream,8254578}, various queue management techniques such as the Last-Come-First-Served (LCFS) protocol \cite{Updateswithqueues}, and models that allow packet discards \cite{packetmanagement,9517796,8406966,Rajai} and deadlines \cite{deadline}. A partial list of studies that seek to compute or to minimize the age under energy or link constraints is \cite{8123937,7283009,8437904,7308962,8254156,8733195,8335672,8006703,8437547,8406974,8422086,8437573,distGaussian,8377368,wang2018skip,Inan2206:Age}. For a comprehensive survey over the AoI literature, see \cite{yates2020ageSurvey}; and for a tutorial, see \cite{tutorial}.  %\cite{Kaul,Wiener_estimate,OU_estimate,8006504,8849695,8886357,9086190,9099557,9333607,AoI_Tutorial,AoI_distribution,AoI_errored,AoI_markov_channel,AoI_multiclass,AoI_multiclass_preemptive,AoI_mutual_info,AoI_nonlinear,AoI_nonlinear2,AoI_packet_management,AoI_through_queues,AoI_trans,Cost_of_delay,Update_or_wait,Vehicular,Yates_Moments}.
	
	Although the works cited above mostly assume error-free transmissions, some others take into account that packets get lost or erased while passing through the network. In \cite{AoIerror}, the AoI is studied in a model where transmissions are error-prone. One may notice that a simple method to combat erasures is to send the packet repeatedly. A more complicated method could send coded packets, which are then to be conveyed through the erasure channel. A list of works concerning coded transmissions with feedback is \cite{8761668,8362277,9440981,8849636,7925903,8006541,8445909,8006504}. An example of a study assuming no feedback is \cite{ElieErasure}, where the authors find the optimal coding strategy. In Section \ref{sec:erasure}, we will discuss that the optimal strategy for our discrete-time model is also optimal for a model where communications take place over an erasure channel with feedback.
	
	AoI has also found place in stochastic control literature. For instance, in \cite{9029447}, a tradeoff between the information staleness and performance of a Linear Quadratic Regulator (LQR) is illustrated. In \cite{8814627}, freshness is taken as basis for an algorithm devised for distributed tracking of a linear system. These works are aligned in the sense that using the fresh data to track and control a system might improve the performance. This is because of the Markovian nature of the processes that are tracked --- the next state depends only on the freshest data. However, this may not always be the case as the freshest may not be the most important. This observation is in line with some further studies. For instance in \cite{Wiener_estimate}, a problem of generating timely updates in a remote estimation setting has been proposed. The authors have investigated the Mean-square-optimal and AoI-optimal strategies for the estimation of a Wiener process through a queue and concluded that they are different; consequently demonstrating a tradeoff between freshness and importance. In \cite{OU_estimate}, the authors generalized the settings to include an Ornstein--Uhlenbeck process. In \cite{distGaussian}, a tradeoff between timeliness and distortion is shown for the case of estimation through a Gaussian channel in a power constrained setting. There also has been several works on integrating the notion of different data importance and timeliness, e.g., by introducing non-linear cost to stale data \cite{AoI_nonlinear,Cost_of_delay}, by considering separate data streams of different priorities \cite{8437591,8886357}, or by modeling the distortion as a decreasing function of the service time \cite{8988940}.
	
Our setup contains flavors from the above approaches, yet it features novel aspects. For instance, the resource constraint is imposed by an external scheduler, giving the sender turns to speak. Hence, as opposed to the works \cite{Wiener_estimate,OU_estimate}, the sender cannot decide when to send; but it rather decides what to send. Furthermore, we adopt the view that data is formed into packets of different importance levels, e.g., packets containing abnormal levels of coolant temperature in a nuclear plant could be classified as important. Consequently, the distortion metric we propose depends on whether the packets are received or not, and the accumulated importance levels of the missed data constitutes our distortion metric. 

	\section{Notation}
	Random variables are denoted with uppercase letters (e.g., $X$); and vectors are denoted with boldface letters (e.g., $\bs{b}$). Sets are denoted with script-style letters (e.g., $\cV$). $l(\bs{b})$ is the length of a vector $\bs{b}$, and $b_i$ is its $i^\text{th}$ element. For vectors $\bs{b}$ and $\bs{b}'$, $\bs{b}\| \bs{b}' := [b_1,b_2,\ldots,b'_1,b'_2,\ldots]$ is the concatenation of the two vectors. $\bs{b}_{\geq i} := [b_i,\dots,b_l]$ is segment of $\bs{b}$ from its $i^\text{th}$ element until the end; and $\bs{b}_{i}^j : = [b_i,\dots,b_j]$ is the segment between its $i^\text{th}$ and $j^\text{th}$ elements, $\bs{b}^i := \bs{b}_1^i$ . $\bs{b}'$ is a suffix of $\bs{b}$ if there exists an $i > 1$ such that $\bs{b}' =\bs{b}_{\geq i}$. If $\bs{b}' = \bs{b}_{\geq i}$ is suffix of $\bs{b}$, then $\bs{b}\setminus\bs{b}' = \bs{b}^{i-1}$.  For $a,b \in \mathbb{R}$, $a\wedge b := \min\{a,b\}$, and  $a \vee b:= \max\{a,b\}$.
	
	\section{Problem Definition}\label{sec:problem_definition}
	%We propose the following model. The purpose is to keep the model simple enough to analyze, and to allow the model to capture both the timeliness and distortion concepts. The model consists of the data to be sent, the sender-receiver pair, and the channel in between.
	
	In this section, we describe our discrete-time model in terms of the data to be conveyed, the sender-receiver pair with their respective communication protocol, and the channel in between.
	
	We assume that the data is formed into packets, and at each time instant $t$, a new packet arrives to the sender. The packet payloads originate from a set of finite elements $\cX$, and the probability of a payload taking a particular value is time-invariant and independent of the past. Consequently, the data is an independent and identically distributed (i.i.d.) process $\{X_t\}_{t \in \mathbb{N}}$. The sender observes $X_t$ at time $t$ and keeps $X_t$ in its buffer. 
	
	The communication protocol is as follows:  The sender is allowed to speak at times $T_1, T_2, \ldots$. The process $\{T_i\}_{i \in \mathbb{N}}$ is independent of the process $\{X_t\}_{t \in \mathbb{N}}$, and has the property that the interspeaking times $Z_i := T_i-T_{i-1}$ are i.i.d.. Moreover, we assume that $Z_i$'s are strictly positive and square integrable, i.e., $\Pr(Z_i > 0) = 1$ and $E[Z_i^2] \leq \infty$. An example of such a random variable could be a geometric random variable with $\Pr(Z_i = t) = p(1-p)^{t-1}$ for $t \geq 1$. The speaking process $\{T_i\}_{i \in \mathbb{N}}$ is inspired by MAC layer protocols where each sender is assigned time slots to speak. When the sender is given a turn to speak, i.e., at each $T_i$, it selects a packet from its buffer with timestamp $S_i \leq T_i$ and forwards $X_{S_i}$. Once $X_{S_i}$ is forwarded, we restrict the sender to not send a packet with timestamp less than $S_i$ at the subsequent speaking times $T_{i+1}, T_{i+2},\ldots$. Note that such restriction results in $S_i < S_{i+1}$. The increasing sequence $\{S_i\}_{i \geq 0} =: \bs{S}$ is henceforth referred as the `selection process'.

	Transmissions between the sender and the receiver are noiseless and zero-delay. Hence, by time $t$, the receiver has observed $X_{S_i}$ for every $i$ such that $T_i < t$. We also suppose that the packets are formed to contain timestamps, i.e., the packet containing $X_{S_i}$ also contains the information that it was generated at time $S_i$. Consequently, at time $t$, the receiver is able to reconstruct the data as $Y_j(t) = X_j$ if $X_j$ is among its observation up to time $t$; otherwise it sets $Y_j(t) = \?$.
	
%	At time $t$, the receiver has observed $X_{S_i}$ for every $i$ such that $T_i < t$. Consequently at time $t$, it can reconstruct the data as $Y_j(t)= X_j$ or $Y_j(t)=\text{?}$ depending if $X_j$ is among the receiver's observation up to time $t$ or not.
	
	%	The receiver keeps track of a reconstruction sequence $Y_{1}(t),\dots,Y_{t}(t)$ for all $t$ with $Y_i(t) \in \cY$. As soon as it receives a packet, i.e. when it receives a packet with index $S_i$ at time $T_i$, it finalizes its estimate of packet $S_i$ and its past. In other words, it sets $Y_{j}(t) = Y_{j}(T_i)$ for $t \geq T_i$ and $j \leq S_i$. 
	
	At this point, we have described our model. Now, we introduce the appropriate distortion and timeliness metrics to study their tradeoff. Specifically, given $d: \cX \times \cX \cup \{\?\} \to \mathbb{R}_{\geq 0}$, with
	\begin{equation}
	d(x,x) = 0 \text{ and } d(x,\?) = :v(x),
	\end{equation}
	and given a selection procedure $\bs{S}$, define
	\begin{equation}
	D_t^{(\bs{S})} := \frac 1 {t} \sum_{i=1}^{t} d(X_i,Y_i(t)) \text{\ and\ } D^{(\bs{S})}:= E\bigg[\limsup_{t \to \infty}D_t^{(\bs{S})}\bigg].
	\end{equation}
	With an analogy to rate-distortion theory, observe that $D_t^{(\bs{S})}$ quantifies average distortion between the source and its reconstruction. $D^{(\bs{S})}$ is the expected long-term average distortion.
	
Timeliness of information is quantified with the well-studied AoI metric. Namely, with $i(t):= \sup\{i \geq 0: T_i < t\}$, $T_0=S_0 = 0$; define for all $t > 0$,
		\vspace*{-7pt}
	\begin{equation}\label{label: delta_renewal}
	\Delta_t^{(\bs{S})} := t-S_{i(t)},\text{\ and\ }\Delta^{(\bs{S})} := E\bigg[ \limsup_{t\to \infty}{\frac 1 t \sum_{\tau=1}^t\Delta_\tau^{(\bs{S})}}\bigg]\text.
	\end{equation}
	$\Delta_t^{(\bs{S})}$ is usually referred as the instantaneous age; and similar to $D^{(\bs{S})}$, $\Delta^{(\bs{S})}$ is the expectation of the long-term average age.
	
	%	Since we assumed that the $X_i$s are i.i.d., the packets that are successfully transmitted have no information about the packets that are not transmitted (or erased). Hence, we represent the receiver's 'best guess' for the erased packets with the symbol $? := \arg\min_{y \in \cY} E[d(X,y)].$
	
	Note that $Y_i(t)$ can be either equal to $X_i$ or to `$\?$'. Therefore, specifying only $d(x,x)$ and $d(x,\?)$ --- which is readily determined by $v(x)$ --- is sufficient to evaluate $D^{(\bs{S})}$. As a consequence, the sender may base its selection $S_i$ on $\bs{V}^{T_i}$,
%	As mentioned, at time $T_i$, the sender selects $S_i$ on the basis of its buffer $X_{S_{i-1}+1}^{T_i}$ --- recall that the buffer is cleaned of the past data $X_1^{S_{i-1}}$ --- and since its strategy $\bs{S}$ is evaluated with the pair $(\Delta^{(\bs{S})},D^{(\bs{S})})$, it may as well base its selection on 
%	$$\bs{B}_i := V_{S_{i-1} +1}^{T_i},$$ 
	where $V_t := v(X_t)$. 	Therefore, the selection $S_i$ is a map $S_i: \cV^{T_i} \to \{S_{i-1}+1,\ldots,T_i\}$ with $\cV := \{v(x):{x \in \cX}\}$. Intuitively, $V_i$ represents an importance score for the packet $i$; high $V_i$ is interpreted as the content having high importance and not sending it incurs a high penalty --- this interpretation is also consistent with a model where some arrivals are prioritized. Observe that the structure of the problem stays the same if all elements of $\cV$ are multiplied by a positive constant. If $\cV$ does not contain $0$, then without loss of generality one can assume that the minimum element in $\cV$ is $1$ and it is an ordered set as $1 = v_1 < v_2 < \ldots < v_{|\cV|} := v_{\max} < \infty$.
	
	Now that we have the full description of the setting, we aim to characterize the achievable region of $(\Delta^{(\bs{S})},D^{(\bs{S})})$ pairs. We attempt to characterize this region in the sequel and conclude this section with a few remarks.
	
%	
%	 so the distortion metric defined as above, specifying only $d(x,x)$ and $d(x,?)$ is sufficient to evaluate $D$. At time $t=T_i$, the transmitter needs to choose $S_i$ on the basis of $X_1^t$.  Since the transmitter's strategy is evaluated by $(\Delta,D)$, the transmitter may as well base its choice on $V_1^t$ with $V_t=v(X_t)$. Intuitively, $V_i$ represents an importance score for the packet $i$; high $V_i$ means the content has high importance and not sending it incurs a high penalty. Observe that the structure of the problem stays the same it all the elements in $\cV$ are multiplied by a positive constant. If $\cV$ does not contain $0$, then without loss of generality one can assume that the minimum element in $\cV$ is $1$ and it is an ordered set as $1 = v_1 < v_2 < \ldots < v_{|\cV|} := v_{\max} < \infty$.
%	
%	Given the description of the model, the main question now is: What are the achievable $(\Delta,D)$ pairs? We attempt to answer this question in the next section and conclude this section with a few remarks.
	
	\begin{itemize}
		\item[(i)]The model we propose is reminiscent of a remote estimation problem of a discrete-time stochastic process through a discrete-time queue. However, we require that the sender sends a packet exactly at speaking times, which is equivalent to force the sender to send as soon as the queue is idle in a discrete-time queueing setting. In \cite{Update_or_wait} and \cite{Wiener_estimate}, it is shown that the optimal policies need not be of this type. This makes our problem different and allows us to make the relaxation that $S_i$ need not be stopping times.
%		\item[(ii)] A more sophisticated receiver could try to reconstruct the missing $X_t$'s from the $X_{S_i}$ it has observed.  This may be possible even when the process $(X_t)_{t \in \mathbb{N}}$ is i.i.d., if the transmitter chooses $S_i$ appropriately; e.g., by favoring $S_i$'s for which $X_{S_i} = X_{S_i-1}$.  Our formulation does not take into account such receivers.
		\item[(ii)]If $0 \in \cV$, there are multiple interpretations. $V_t = 0$ can be interpreted as either the data is totally trivial (need not be reconstructed), or interpreted as the source having not generated data at time $t$. The second interpretation allows us to model a source which generates data sporadically. Now there is the question of allowing $X_t$ to be sent or not. Our model allows sending of $X_t$, i.e., in the second interpretation, informs the receiver that there has not been any data generated by the source, and $\Delta_t$ decreases accordingly. The reduction of $\Delta_t$ can be avoided by appropriate reformulation --- to be discussed in Remark \ref{rem:no_data_sent}.
	\end{itemize}

\section{The Age-Distortion Tradeoff}\label{sec:age_dist_tradeoff}

\subsection{Markov Decision Problem Formulation as a Lower Bound}\label{sec:mdp_formulation}

%Recall that at time $T_i$ and given $(S_i,V_1^{T_i})$, $V_1^{T_{i+1}}$ is independent of the past. This is because the importance values $V_t$'s and the interarrival times $Z_i$ are i.i.d. and also independent of each other. Also recall that the buffer is cleaned of $X_1^{S_{i-1}}$. Hence, the only relevant \begin{equation}
%\bs{B}_{i} := V_{S_{i-1}+1}^{T_i}.
%\end{equation}
%The assumption that $S_0 = 0$ implies that the buffer is empty just after $t = 0$. Note that even though the sequence of selected packets $\{S_n\}_{n \in \mathbb{N}}$ remains the same, $\Delta$ and $D$ may change depending on the initial buffer state in general. The sender's strategy is completely specified by the policy $\bs{B}_i \mapsto S_i$, identifying the data to be transmitted.

To study the age-distortion tradeoff, we study the family of weighted costs $\eta\Delta^{(\bs{S})} + D^{(\bs{S})} $ for $\eta > 0$. It is known that the boundary of the achievable $(\Delta^{(\bs{S})}, D^{(\bs{S})})$ region can be characterized by studying this family. We seek to obtain a tractable lower bound for $\eta\Delta^{(\bs{S})} + D^{(\bs{S})}$, and then we further optimize this lower bound over $\bs{S}$. First, we derive a simpler expression for $\Delta^{(\bs{S})}$. Observe that
\begin{equation}
\limsup_{t\to \infty}{\frac 1 t \sum_{\tau=1}^t\Delta_{\tau}^{(\bs{S})}} = \limsup_{i \to \infty} \frac{\sum_{j=1}^iQ_j^{(\bs{S})}}{\sum_{j=1}^iZ_j}
\end{equation}
\centerline{where}
\begin{equation}
Q_j ^{(\bs{S})}:= (T_j-S_j)Z_{j+1} + \frac{Z_{j+1}(Z_{j+1}+1)} 2.
\end{equation}
Since $\displaystyle\lim_{i \to \infty}\textstyle \frac 1 i \sum_{j=1}^i \frac{Z_{j+1}(Z_{j+1}+1)} 2 = \frac 1 2 E[Z_1(Z_1+1)] =: \nu$ and $\lim_{i \to \infty} \frac 1 i \sum_{j=1}^i Z_j = E[Z_1] =:\mu$ with probability 1 by the law of large numbers, we obtain
\begin{equation}
\Delta^{(\bS)} = E\bigg[\frac 1 {\mu}\limsup_{i \to \infty} \frac 1 i \sum_{j=1}^i(T_j-S_j)Z_{j+1} + \frac \nu \mu\bigg].
\end{equation}
Note that $\Delta^{(\bs{S})}$ cannot be smaller than $\nu/\mu$.  We subtract $\nu/\mu$ to obtain the excess age, given by 
\begin{equation}\label{eqn:delta_excess}
\Delta_e^{(\bS)} :=  E\bigg[\frac 1 \mu \limsup_{i \to \infty} \frac 1 i \sum_{j=1}^i(T_j-S_j)Z_{j+1}\bigg]
\end{equation}
and determine the feasible $(\Delta_e^{(\bs{S})}, D^{(\bs{S})})$ pairs. With the same reasoning as above, we study the family $\eta \Delta_e^{(\bS)} + D^{(\bS)}$. When the selection process $\bS$ satisfies a certain square-integrability condition, we can find alternative expressions for $\Delta_e^{(\bS)} $ and $D^{(\bS)}$.

\begin{theorem}\label{thm:simple_age_D} Let $\cS_2$ be the set of selection processes $\bS$ with $\sup_i E[(T_i-S_i)^2] < \infty$. Then for any  $\bS \in \cS_2$,
		\begin{align}
		\Delta_e^{(\bS)}  &=  E\bigg[\limsup_{i \to \infty} \frac 1 i \sum_{j=1}^i(T_j-S_j)\bigg],\label{eq:simple_age}\\
		D^{(\bS)}  &= \frac 1 \mu E\bigg[\limsup_{i \to \infty} \frac 1 i \sum_{j=1}^i D(\bs{V}^{T_j},S_{j-1},S_{j})\bigg]\label{eq:simple_D},
		\end{align}
		\centerline{where}
		\begin{equation}
		D(\bs{V}^{T_j},S_{j-1},S_{j}) := \sum_{j' = S_{j-1}+1}^{S_j-1} V_{j'}
		\end{equation}
		is the penalty incurred by skipping the portion $[V_{S_{j-1}+1},\ldots,V_{S_j-1}]$ of $\bs{V}^{T_j}$.
\end{theorem}

\begin{proof}See Appendix \ref{app:thm1,2}.
\end{proof}

	At this point, we would like to eliminate some of the non-optimal selection processes. More specifically, we show that if a packet with importance value $\vmin$ and with timestamp less than $T_i$ is selected at time $T_i$, then one can find another selection process which performs at least as well.
	
	\begin{lemma}\label{lem:min_importance}
		Consider a selection process $\bs{S}$ with $S_{i_0} < T_{i_0}$ and $V_{S_{i_0}} = \vmin$ for some $i_0$. Then one can find another process $\tilde {\bs{S}}$ with $V_{\tilde S_{i_0}} > \vmin$ and such that $\Delta_e^{(\tilde \bS)} \leq \Delta_e^{(\bS)}$ and $D^{(\tilde \bS)} \leq D^{(\bS)}$.
	\end{lemma}

\begin{proof}See Appendix \ref{app:min_importance}.
\end{proof}

Lemma \ref{lem:min_importance} helps restrict the search space for possibly optimal selection processes. The processes we study in the sequel will not select a minimum-importance packet if it has not arrived exactly at the speaking time. Let us denote the class of such selection processes as $\cS_2'$. 

A search for an optimal strategy based on the expressions in Theorem \ref{thm:simple_age_D} seems to be a complex task. We aim to obtain a further lower bound for $\eta \Delta_e^{(\bS)} + D^{(\bS)}$ which turns out to be an entity that is amenable for analysis. This lower bound is given in

\begin{theorem}\label{thm:reverse_fatou} Define 
	\begin{equation}\label{eq:onestep_cost}
	J_i(\eta)^{(S_1^{T_i})} :=  E\bigg[\sum_{j=1}^i \frac 1 \mu D(\bs{V}^{T_j},S_{j-1},S_{j}) + \eta(T_j-S_j) \bigg]
	\end{equation}
	\centerline{and}
	\begin{equation}
	J(\eta)^{(\bS)} := \limsup_{i\to\infty} \frac 1 i J_i(\eta)^{(S_1^{T_i})}.
	\end{equation}
	Then, for any $\bs{S} \in \cS_2'$, $ J(\eta)^{(\bs{S})} \leq D^{(\bs{S})} + \eta \Delta_e^{(\bs{S})}$.
\end{theorem}
\begin{proof}
	See Appendix $\ref{app:thm1,2}$.
\end{proof}
A straightforward implication is
		\begin{equation}\label{eq:infimum}
J^*(\eta) := \inf_{\bS \in \cS_2'}  J(\eta)^{(\bs{S})} \leq \inf_{\bS\in \cS_2'}  \big(D^{(\bs{S})} + \eta \Delta_e^{(\bs{S})}\big).
\end{equation}
As we shall see later in Section \ref{sec:exact_buffer}, it turns out that the inequality \eqref{eq:infimum} is indeed an equality. The key reason to introduce $J(\eta)^{(\bs{S})}$ is that its infimization can be formulated as a MDP, which we do next. This requires identifying the states and the actions of the MDP and verifying that (i) the distribution of the next state and (ii) the one-step cost depend only on the current state-action pair. We claim that the buffer content at the $i^{\text{th}}$ speaking time
\begin{equation}
\bs{B}_i := \bs{V}_{S_{i-1}+1}^{T_i},
\end{equation}
and the number of packets to be dropped $A_i := S_i - S_{i-1}$ constitutes this state-action pair.
To see this, note that given $\bs{B}_i$ and $A_i$, (i) the next buffer content is independent of the past; and (ii) the one-step cost 
\begin{equation}
\frac 1 \mu D(\bs{V}^{T_i},S_{i-1},S_{i}) + \eta(T_i-S_i)
\end{equation}
is a function only of $\bs{B}_i$ and $A_i$. This is because the above is equal to
\begin{equation}\frac 1 \mu \sum_{k=1}^{s-1}b_k + \eta(l(\bs{b})-s); \end{equation}
with $\bs{b} = \bs{B}_i$ and $s = A_i$. Since (i) and (ii) are satisfied, the problem is indeed an MDP whose state at instant $i$ is the buffer content $\bs{B}_{i}$; and whose action at instant $i$ is $A_i$. The assumption $S_0 = 0$ allows us to choose $\bs{B}_0$ as an empty buffer.

%From the MDP literature, e.g. \cite{Puterman}, we know that if $\{S_i\}_{i\geq 0}$ is a stationary policy, then $\bs{B}_i \mapsto S_i$ for all $i$; and if $\{S_i\}_{i\geq 0}$ is history-dependent, then $(\bs{B}_1,  \ldots,\bs{B}_i) \mapsto S_i$. Any selection process in the latter case can also be based on $\bs{V}^{T_i}\mapsto S_i$, as $\bs{V}^{T_i}$ and $\bs{B}_1,  \ldots,\bs{B}_i $ carry the same amount of information. Thus, the history-dependent selection processes in our MDP formulation coincide with the selection processes described in Section \ref{sec:problem_definition}.

 The formulation above is an infinite-horizon average-cost MDP \cite{Bertsekas} with states $\bs{b} \in \cV^* := \cup_{k=1}^\infty\cV^{k}$; and the set of possible actions for a state $\bs{b}$ is given by $s \in \{1,\ldots, l(\bs{b})\}$, where $l(\bs{b})$ is the length of the buffer $\bs{b}$. Consequently, the sender chooses the packet with timestamp $S_i = T_i + s - l(\bs{b})$ at time $T_i$. Observe that setting $s= l(\bs{b})$ corresponds to the selection of the freshest packet whereas setting $s= 1$ corresponds to the selection of the oldest packet in the buffer.  Furthermore, since we are interested in the selection procedures in $\cS_2'$, $s \neq l(\bs{b})$ only if $b_{s} > \vmin$. 
 
 In general, the optimal policies of an MDP need not be stationary (the current action is a deterministic function of the current state). We will show in Section \ref{sec:exact_buffer} that for our problem, the optimal policy is indeed stationary. When we consider a stationary selection process $\bs{S}$, we explicitly write the argument $(\bs{b})$. Let us recall

\begin{definition}[Unichain policy, \cite{Bertsekas}] If a stationary policy $s(\bs{b})$ induces a Markov chain with a single recurrent class and a possibly empty set of transient states, it is called unichain.
\end{definition}

In an average-cost dynamic programming setting, we evaluate a unichain policy $s(\bs{b})$, $\bs{b} \in \cV^*$ by solving the linear system with unknowns $h(\bs{b})$, $\bs{b} \in \cV^*$ and $\lambda$; given by 
\begin{equation}\label{eq:linear_solve}
h(\bs{b}) + \lambda = \frac 1 \mu \sum_{k=1}^{s(\bs{b})-1}b_k + \eta(l(\bs{b})-s(\bs{b})) + E[h(\bs{b}_{\geq s(\bs{b})+1} \| \bs{V}^{Z})],
\end{equation}
where $Z$ is the next interspeaking time, $\bs{V}^{Z}$ is a vector of i.i.d. $V$'s of length $Z$, $h(\bs{b})$ is called the relative value of state $\bs{b}$, and $\lambda$ is the average cost induced by this policy. Since \eqref{eq:linear_solve} determines $h(\bs{b})$ up to an additive constant, we take a reference state --- see \cite[Chapter 4]{Bertsekas} --- as one of the $\bs{b} \in \cV^*$ and we set $h(\bs{b}) = 0$. We also note that for a unichain policy, the linear system given by \eqref{eq:linear_solve} has a unique solution \cite{Bertsekas}.

\begin{remark}\label{rem:no_data_sent}
	To cover the case where $v=0$ is interpreted as the source having not generated any data and $\Delta_t$ should not decrease upon the sending of $v = 0$; one can proceed as follows: The set of possible actions for a state $\bs{b}$ is extended to $\{0, 1,\ldots, l(\bs{b})\}$, and $s> 0$ only if $b_{s} > 0$. That is, the sender is allowed to choose only the packets with $v  > 0$. Also note that for the all-zero buffer, the only possible action is to set $s= 0$, i.e., nothing has arrived since the last selection and hence there is nothing to send. Observe that in this case $\Delta_t$ does not drop. 
\end{remark}

Although we have characterized a lower bound based on a MDP formulation, the formulated problem has a countably infinite state space, $\cV^*$. It is known that for this class of problems, analysis of optimal policies become formidably complex in general. Moreover, it is not certain that a stationary policy attains the infimum in \eqref{eq:infimum}. The complexity of this problem leads us to consider a finite-state modification of the problem; and the next section is devoted for this modified version.

%To be able to write simple expressions for $\Delta_e$, $D$ and relate those to one-step costs of a dynamic programming problem, we make a technical assumption that $\sup_i E[(T_i-S_i)^2] < \infty$. We give these simple expressions in the theorem below. The proof is found in Appendix \ref{pf:thm1}.
%
%\begin{theorem}\label{thm:simple_age_D} For policies with $\sup_i E[(T_i-S_i)^2] < \infty$,
%	$$\Delta_e =  E\bigg[\limsup_{i \to \infty} \frac 1 i \sum_{j=1}^i(T_j-S_j)\bigg],$$
%	$$D = \frac 1 \mu E\bigg[\limsup_{i \to \infty} \frac 1 i \sum_{j=1}^i D(\bs{B}_{j},S_{j-1},S_{j})\bigg].$$
%	Furthermore,
%	$ J(\eta)^{(\bs{S})} \leq D^{(\bs{S})} + \eta \Delta_e^{(\bs{S})}$ and for stationary policies the equality holds. Hence, solving the optimization problem in \eqref{eqn:optimization} gives a lower bound in general.
%\end{theorem}

%Given the current buffer content $\bs{B}_{j}$, the next state depends only on $\bs{B}_{j}$ and $S_j$. Hence, the optimization problem \eqref{eqn:optimization} can be formulated as dynamic programming where the states of subject MDP are described by the buffer content $\bs{B}_{j}$. More precisely, our formulation is an infinite-horizon average-cost dynamic programming problem with states $\bs{b} \in \cV^*$ in MDP terminology \cite{Bertsekas}. For analyzing this dynamic programming problem, let us recall:
%
%\begin{definition}[Unichain policy, \cite{Bertsekas}] If a deterministic stationary policy $s(\bs{b})$ induces a Markov Chain with a single recurrent class and a possibly empty set of transient states, it is called unichain.
%\end{definition}
	
	\subsection{Policy Iteration with a Truncated State Space}\label{sec:truncated}
	We now consider a finite-state version of the problem where the sender forgets the packets that have arrived more than $K$ time slots ago. That is, the current buffer content at time $T_i$ becomes
	\begin{equation}
	\bs{B}_i := \bs{V}_{(S_{i-1}+1)\vee (T_i-K+1)}^{T_i}.
	\end{equation}
	Consequently, the buffer length is limited to at most $K$, and the state space becomes finite. Denote this finite state space by $\cV^{\leq K} := \cup_{l\leq K} \cV^l$. One may notice that by restricting the state space, we might not attain the infimal value $J^*(\eta)$. However, as we shall see later in \ref{sec:exact_buffer}, the optimal policy of the original \emph{infinite} state space problem will base its decisions only on a bounded buffer. Thus, we do not lose optimality provided that $K$ is large enough.
	
	A slight modification of \eqref{eq:linear_solve} is enough to obtain the linear system whose solution yields the relative values and the average cost. First, observe that if the buffer state $\bs{b}$ has length more than $K$, i.e., if $l(\bs{b}) > K$, then $h(\bs{b})$ can be replaced with
	\begin{equation}
	\sum_{k = 1}^{l({\bs{b}})-K} b_k + h(\bs{b}_{\geq l(\bs{b})-K+1}).
	\end{equation}
	as the first $l(\bs{b}) -K$ terms will be forgotten in the truncated problem.
%	Also observe that in this case the expectation on the RHS of \eqref{eq:linear_solve} involves terms $h(\tilde{\bs{b}})$ with $l(\tilde{\bs{b}})$ which might be more than $K$. Such terms should be evaluated as  $h(\tilde{\bs{b}}) = \sum_{k = 1}^{l(\tilde{\bs{b}})-K} \tilde{b}_k + h(\tilde{\bs{b}}_{l(\tilde{\bs{b}})-K+1}^{l(\tilde{\bs{b}})})$.
	Let $p_i := \Pr(Z = i)$ and $q_i := \Pr(Z \geq i)$. Then the linear system of equations to evaluate a unichain stationary policy $s(\bs{b})$, $\bs{b} \in \cV^{\leq K}$ becomes
	\begin{align}\label{eqn:linear_solve_K}
	h(\bs{b}) + \lambda &= \frac 1 \mu \sum_{k=1}^{s(\bs{b})-1}b_k + \eta(l(\bs{b})-s(\bs{b}))\mbox{}+ \frac 1 \mu \sum_{k= s(\bs{b})+1}^l b_kq_{K+k-l} + \frac {E[V]} \mu E[(Z-K)^+]\notag\\
	&\phantom{=}+ \sum_{k = 1}^{K-l+s(\bs{b})}p_k E[h(\bs{b}_{\geq s(\bs{b})+1} \| \bs{V}^k)] + \sum_{k = K-l+s(\bs{b})+1}^{K-1}p_k E[h(\bs{b}_{\geq k-(K-l)+1} \|\bs{V}^k)]+ q_{K}E[h(\bs{V}^K)]\\
	&=: C_{\bs{h}}(\bs{b},s(\bs{b})),\label{eq:defn_c}
	\end{align}
	with $h(v_{\min}) = 0$. This choice of the reference state also sets $h(b) = 0$ for $b \in \cV$. Namely, for buffers that contain only one packet, the relative value will be equal to zero.
	
	At this point, we have obtained the policy evaluation method for our finite-state problem. A first attempt could be to find the optimal policies numerically. We consider the well-known policy iteration algorithm \cite{Bertsekas}. A brief description is given in Algorithm \ref{alg:policy_iter}.
	\begin{algorithm}\label{alg:policy_iter}\caption{Policy iteration}
		Start with the stationary policy $s^{(0)}(\bs{b}) = l(\bs{b})$.\\
		Evaluate $s^{(i)}(\bs{b})$ according to \eqref{eqn:linear_solve_K} to find $h^{(i)}(\bs{b})$, $\bs{b} \in \cV^{\leq K}$ and $\lambda^{(i)}$.\\
		$\text{For all }\bs{b} \in \cV^{\leq K} $, set $$ s^{(i+1)}(\bs{b}) = \arg\min_{\substack{s:s \leq l(\bs{b})\\ b_{s} \neq \vmin \text{ if } s < l(\bs{b})}} C_{\bs{h}^{(i)}}(\bs{b},s)$$\\
		If $s^{(i+1)}(\bs{b}) = s^{(i)}(\bs{b})$ for all $\bs{b} \in \cV^{\leq K}$, terminate. Else go to step 2.
	\end{algorithm}
	
	Although the number of states is now finite, it is not clear that the policy iteration algorithm yields a unichain policy. We shall show this in the following lemma.
	
	\begin{lemma}\label{lem:unichain} If the buffer size is limited to $K$, the policy iteration terminates with an optimal unichain policy in $\cS_2'$.
	\end{lemma}
	\begin{proof}
		Consider all truncated policies, e.g., non-stationary, history dependent but can only choose the most recent $K$ packets in the buffer. Since $\{V_i\}_{i \geq 0}$ is an i.i.d. process, any policy eventually reaches a state consisting of only $\vmin$'s and must choose the most recent packet. Hence, the buffer must be eventually renewed for any policy in $\cS_2'$ and as a result, there must be a single recurrent class. Therefore, there exists an optimal stationary and deterministic strategy that is unichain and this policy can be found with the policy iteration algorithm  \cite{Bertsekas}.
	\end{proof}

\begin{remark}\label{remark:perturb} Intuitively, step 3 of the above algorithm modifies $s^{(i)}(\bs{b})$ in the following way: Consider two processes starting at the state $\bs{b}$. The first one is iterated with respect to $s^{(i)}$, whereas the second one is iterated with a different action $\tilde s(\bs{b})$ at the first step and with $s^{(i)}$ subsequently. Now consider the expected accumulated costs of these two processes until they reach the same state. If the second process has a smaller expected accumulated cost, changing all $s^{(i)}(\bs{b})$ to $\tilde s(\bs{b}) = s^{(i+1)}(\bs{b})$ results in a better policy; otherwise try another $\tilde s(\bs{b})$.
\end{remark}
	
	\subsection{The Exact Buffer Size for an Optimal Policy}\label{sec:exact_buffer}
	
	Truncating the state space restricts the actions that may be taken. Therefore, in general, the infimum in \eqref{eq:infimum} may not be attained with a truncated buffer. As we have said in the previous section, it turns out that this is not the case for our problem and the infimum is indeed attained with a finite buffer size. In this section we quantify this buffer size.
	
	First, consider the policy $s(\bs{b}) = l(\bs{b})$ for all $\bs{b} \in \cV^{\leq K }$, i.e., always send the most recent packet in the buffer. One can observe that this policy induces a Markov chain with only $|\cV| = 2$ states regardless of $K$. We shall now show that this policy is optimal for $\eta$ above some threshold $\eta_{\max}$.
	
	\begin{lemma}\label{lemma:eta_critical} For $\eta \geq \eta_{\max} := \frac 1 \mu(v_{\max} - v_{\min})$ and for any $M \geq 1$, the optimal policy among $\cV^{\leq M}$ is $s(\bs{b}) = l(\bs{b})$; which can be implemented with a buffer size of 1.
	\end{lemma}
	\begin{proof} We show that the policy $s(\bs{b}) = l(\bs{b})$ remains unchanged under policy iteration. Start the policy iteration with $s^{(0)}(\bs{b}) = l(\bs{b})$. Recalling Remark \ref{remark:perturb}, we will show that perturbing the policy at initial step cannot decrease the expected accumulated cost until the original and perturbed processes coincide. Assume the perturbed action is $\tilde s(\bs{b}) = l(\bs{b}) - k$ for a $k > 0$. Notice that the two processes will coincide immediately at the next step and the difference of the accumulated costs will be $k\eta - \frac 1 \mu (b_{l(\bs{b})-k}-b_{l(\bs{b})})\geq \eta - \frac 1 \mu (v_{\max}-v_{\min}) \geq 0$. Hence the policy remains unchanged.
	\end{proof}
	
	%	\begin{lemma}
	%		For any $K> 1$, there exists a $0 <\eta_K < \eta^*$ such that a buffer size of $K$ is enough for an optimal strategy with respect to the cost function $J^{(\bs{S})}(\eta)$, $\eta \in [\eta_K,\eta^*]$.
	%	\end{lemma}
	%	
	%	\begin{proof} It is immediate from \lemref{lemma:eta_critical} that if such $\eta_K$ exists, it cannot be greater than $\eta^*$.
	%	\end{proof}
	
Considering the truncated state space $\cV^{\leq K}$, we give some properties of optimal policies.
	\begin{property}\label{prop:opt_soln} 
		For an optimal stationary policy $s^*(\bs{b})$, and optimal relative values $h^*(\bs{b})$, the following hold:
		\begin{itemize}
%			\item[(i)] $h^*(\bs{b}) \leq \frac 1 \mu(b_1+\ldots + b_{l-1})$, for all $\bs{b} \in \cV^{\leq K}$
			\item[(i)] For any state $\bs{b} \| \bs{b}'$, either $s^*(\bs{b} \| \bs{b}') = l(\bs{b}) + s^*(\bs{b}')$ or $s^*(\bs{b}\|\bs{b}') \leq l(\bs{b})$.
%			\item[(ii)]$h^*(\bs{b}')\leq h^*(\bs{b})$ for $\bs{b}'$, $\bs{b} \in \cV^l$ if $b_i' \leq b_i$ for all $i \leq l$.
			\item[(ii)] $h^*(\bs{b}') \leq h^*(\bs{b}\|\bs{b}') \leq \frac 1 \mu(b_1+\ldots + b_{l(\bs{b})}) + h^*(\bs{b}')$ for $\bs{b}$, $\bs{b}' \in \cV^{\leq K}$.
%			\item[(iv)] If $s^*(\bs{b}) \neq l(\bs{b})$, then $v_{s^*(\bs{b})} \neq v_{\min}$
%			\item[(v')] Consider a $v' \in \cV$. For $\eta > \frac 1 \mu(v'-v_{\min})$, if $s^*(\bs{b}) \neq l(\bs{b})$, then $v_{s^*(\bs{b})} \neq v'$.
		\end{itemize}
	\end{property}
\begin{proof}
	See Appendix \ref{pf:prop1}.
\end{proof}

To get a sense of how far can the optimal policy go back in time, i.e., to measure how large $T_i-S_i$ can be, it may be informative to consider the following extreme case, which yields a lower bound on the maximal possible value of $T_i-S_i$.

	\begin{lemma}\label{lem:extreme} For the state $\bs{b} = [v_{\max},\underbrace{v_{\min},\ldots,\vmin}_{L-1}]$,
		\begin{equation}
		s^*(\bs{b})  = \begin{cases}1,\quad &\eta \leq \frac 1 \mu\frac{(v_{\max}-v_{\min})}{L-1}\\
		L,\quad & \eta  >  \frac 1 \mu\frac{(v_{\max}-v_{\min})}{L-1}
		\end{cases}.
		\end{equation}
	\end{lemma}
	\begin{proof} Since we work with policies in $\cS_2'$, the two possible actions for this state are either choosing the $v_{\max}$ at the beginning or choosing the $v_{\min}$ at the end. Referring to Remark \ref{remark:perturb}, suppose at iteration $i$ we have $s^{(i)}(\bs{b}) = l(\bs{b})$ and $\tilde s(\bs{b}) = 1$; and we aim to find the difference of accumulated costs until the original and the perturbed processes coincide. Observe that these processes coincide immediately after the first step and the difference will be $\eta(L-1) + \frac 1 \mu (v_{\min}-v_{\max})$. Then, $s^{(i+1)}(\bs{b}) = 1$ if $\eta(L-1) \leq \frac 1 \mu (v_{\max}-v_{\min})$; otherwise $s^{(i+1)}(\bs{b}) = l(\bs{b})$. Note that the difference does not depend on $i$ and hence the statement for $s^{(i+1)}(\bs{b})$ is also true for $s^*(\bs{b})$.
	\end{proof}
	The above lemma therefore gives a necessary buffer size for a possibly optimal policy as it tells that at $\eta = \frac 1 \mu\frac{(v_{\max}-v_{\min})}{L-1}$, the first packet in the buffer given in Lemma \ref{lem:extreme} is chosen by the optimal policy. Hence, attaining the optimal policy requires a buffer size of at least $\lceil\frac 1 \mu\frac{(v_{\max}-v_{\min})}{\eta}\rceil$. Observe that this does not imply that the optimal policy is reached within this particular finite buffer size. Nevertheless, we can prove that this is indeed the case.
	
	\begin{theorem}\label{thm:suff_buffer_size} For $M \geq K(\eta) := \lceil\frac 1 \mu\frac{(v_{\max}-v_{\min})}{\eta}\rceil$, the optimal policy among $\cV^{\leq M}$ is attained by a policy with buffer size $K(\eta)$. Furthermore, if $b_{s^*(\bs{b})} = v_i$, then $l(\bs{b}) - s^*(\bs{b}) < K_i(\eta) := \lceil\frac 1 \mu\frac{(v_{i}-v_{\min})}{\eta}\rceil$ for all $1  \leq i \leq |\cV|$.
	\end{theorem}
	\begin{proof} See Appendix \ref{app:suff_buffer_size}.
	\end{proof}

	\thmref{thm:suff_buffer_size} implies that when the policy iteration terminates, the policy it outputs not only solves the Bellman equation for the state space $\cV^{\leq M}$ for every $M \geq K(\eta)$, but it also solves the Bellman equation for the countable state space $\cV^*$. Since $h^*(\bs{b})$ is finite and $s^*(\bs{b})$ is attained for every $\bs{b}$, a straightforward extension of Proposition 2.1 in \cite[Chapter 4]{Bertsekas} concludes that $s^*(\bs{b})$ is indeed the optimal policy that attains $J^*(\eta)$. Let $\bs{S}^* = \{s^*(\bs{B}_i)\}_{i \geq 0}$ be the random sequence of the actions taken by the stationary and deterministic policy $s^*(\bs{b})$. Recall the inequality 
	\begin{equation}\label{eq:ineq_renewal}
	J^*(\eta) \leq D^{(\bs{S}^*)} + \eta \Delta_e^{(\bs{S}^*)}
	\end{equation}
	given in Theorem \ref{thm:reverse_fatou}. Furthermore, observe that the buffer state process $\{\bs{B}_i\}$ controlled by $\bs{S}^*$ is a renewal process --- this follows from Lemma \ref{lem:unichain}. Hence, by the renewal reward theorem \cite[Theorem 3.6.1]{Ross} we have
	\begin{equation}
	\begin{split}
	\Delta_e^{(\bS^*)}  &=  E\bigg[\limsup_{i \to \infty} \frac 1 i \sum_{j=1}^i(T_j-S_j)\bigg],\\
	&=  E\bigg[\lim_{i \to \infty} \frac 1 i \sum_{j=1}^i(T_j-S_j)\bigg]\\
	&= \lim_{i \to \infty} \frac 1 i \sum_{j=1}^i  E\big[T_j-S_j\big]
	\end{split}
	\end{equation}
	and similarly
	\begin{equation}
		D^{(\bS^*)}  = \frac 1 \mu \lim_{i \to \infty} \frac 1 i \sum_{j=1}^i E\big[D(\bs{V}^{T_j},S_{j-1},S_{j})\big].
	\end{equation}
	The above shows that in fact the inequality \eqref{eq:ineq_renewal} is an equality. Therefore, the optimal policies for the MDP give the tangent lines to the exact boundary curve of the achievable $(\Delta_e,D)$ region. Consequently, by varying $\eta$, this curve can be found. We end this section with the following corollary that summarizes the above results.
	
	\begin{corollary}\label{corr:results}
		The optimal policy among untruncated state space policies is attained with a buffer size $K(\eta)$ and can be found with the policy iteration algorithm run over the state space $\cV^{\leq K(\eta)}$, which returns $J^*(\eta)$. Furthermore, the least upper bound to the family of straight lines $D + \eta \Delta_e = J^*(\eta)$, $\eta >0$ gives the boundary of the achievable $(\Delta_e,D)$ region.
	\end{corollary}
	\subsection{An Efficient Algorithm to Find the $(\Delta_e, D)$ Region}
	 Although one can run the generic policy iteration algorithm to find the tangent lines to the achievable $(\Delta_e, D)$ region, this turns out to be highly inefficient. In this section, we provide an efficient modification of the policy iteration algorithm. The main idea is to exploit the following property. We omit its proof as it is a straightforward extension of Property 1(i).
	\begin{property}\label{prop:algorithm}
		Consider the policy iteration algorithm (Algorithm \ref{alg:policy_iter}). For any $\bs{b}$, the policy $s^{(i)}(\bs{b})$ is either equal to $s^{(i)}(\bs{b}_{\geq 2}) + 1$, or to $1$. Furthermore, if $s^{(i)}(\bs{b}) = s^{(i)}(\bs{b}_{\geq 2}) + 1$, then $h^{(i)}(\bs{b}) = b_1/\mu + h^{(i)}(\bs{b}_{\geq 2})$.
	\end{property}

Property 2 implies that if $s^{(i)}(\bs{b}) \neq 1$, then there must exist a $\bs{b}'$, which is a suffix of $\bs{b}$, and with $s^{(i)}(\bs{b}') = 1$. Consequently,
\begin{equation}\label{eq:cost_to_parent}
h^{(i)}(\bs{b}) = h^{(i)}(\bs{b'}) + \sum_{b \in \bs{b}\setminus \bs{b}'} b/\mu.
\end{equation}
 The above observation leads to an improvement in the policy evaluation stage of the algorithm as we shall see shortly. Denote the set of states $\bs{b}'$ with $s^{(i)}(\bs{b}') = 1$ as $\cB_1$. In light of \eqref{eq:cost_to_parent}, we see that since the relative values for the other states $\bs{b} \notin \cB_1$ can be determined based on the states in $\cB_1$, it is sufficient for the linear system in the policy evaluation step to include the states in $\cB_1$. We observed empirically that $|\cB_1|$ is much smaller compared $ |\cV^{\leq K}|$. As solving a linear system with $n$ variables has $O(n^3)$ complexity, reducing the set of variables to $\cB_1$ results in a significant improvement.

If the algorithm is modified as suggested above, the policy evaluation step gives $h^{(i)}(\bs{b}')$, $\bs{b}' \in \cB_1$, and the average cost $\lambda^{(i)}$. To find other $h^{(i)}(\bs{b})$'s, we refer to equation \eqref{eq:cost_to_parent}, which suggests that the appropriate data structure to represent the states is a tree structure, denoted as $\cT$, where a state $\bs{b}$ has children $\{b \|\bs{b}\}_{b\in \cV}$. That is, $\bs{b}_{\geq 2}$ is the parent of $\bs{b}$, denoted by $\parent(\bs{b})$, and every suffix of $\bs{b}$ is its ancestor. Then the final statement of Property 2 translates into the recursion $h^{(i)}(\bs{b}) = b_1/\mu + h^{(i)}(\parent(\bs{b}))$ for $\bs{b} \notin \cB_1$.

Along with Property \ref{prop:algorithm}, Theorem \ref{thm:suff_buffer_size} also provides simplifications for the search of an optimal policy. Namely, for a state $\bs{b}$, and for any iteration $j$, $s^{(j)}(\bs{b}) \neq s$ for an $s < l(\bs{b}) - K_{b_s}(\eta)$, where $K_{b_s}(\eta) = K_i(\eta)$ if $b_{s} = v_i$. Using the above facts, we are ready to provide a more efficient version of the policy iteration algorithm, Algorithm \ref{alg:mod_policy_iter}, which is tuned for our problem. The for loops over the tree $\cT$ (lines 9 and 12) are in breadth-first manner. Recall that $h^{(i)}(b) = 0$ for $b \in \cV$ and we set $h^{(i)}(\delta) = 0$ where $\delta$ denotes the empty string.
 	
	\begin{algorithm}\label{alg:mod_policy_iter}
		
		\SetKwProg{Init}{Initialize}{}{}
		\SetKwRepeat{Do}{do}{while}
		
		\caption{Efficient Policy Iteration v1($\eta$)}		
		\KwIn{$\eta$}
		\KwOut{$J^*(\eta)$}
		\Init{}{
		$K \gets K(\eta)$\;
		$s^{(0)}(\bs{b}) \gets l(\bs{b})$\;
		$\cB_1 \gets \{v_{\min}\}$\;
		Set $\delta$ as the root of $\cT$ and add the children $\{b\|\bs{b}\}_{b \in \cV}$ for every buffer $\bs{b} \in \cT$ such that $l(\bs{b}) < K$\;
		$i \gets 0$\;
	}
	\Repeat{$\bs{s}^{(i+1)} = \bs{s}^{(i)}$}{
		\tcc{Policy evaluation}
Find $\lambda^{(i)}$, $h^{(i)}(\bs{b}')$, $\bs{b}' \in \cB_1$ by solving \eqref{eqn:linear_solve_K}\;
		\For{$\bs{b} \in \cT\setminus \cB_1$ with $l(\bs{b}) > 1$}{
			$h^{(i)}(\bs{b}) \gets h^{(i)}(\parent(\bs{b})) + b_1/\mu$\;
		}
	\tcc{Policy update}
	$\cB_1 \gets \{v_{\min}\}$\;
	\For{$\bs{b} \in \cT$ with $l(\bs{b}) > 1$}{
		\If{$l(\bs{b}) < K_{b_1}(\eta)$ and $C_{\bs{h}^{(i)}}(\bs{b},1) < C_{\bs{h}^{(i)}}(\parent(\bs{b}),s^{(i+1)}(\parent(\bs{b}))) + b_1/\mu$}{
				$s^{(i+1)}(\bs{b})  \gets 1$\;
				Add $\bs{b}$ to $\cB_1$\;
		}
		\Else{$s^{(i+1)}(\bs{b}) \gets s^{(i+1)}(\parent(\bs{b}))+1$\;
		}
	}
	$i \gets i+1$.}
\Return $\lambda^{(i)}$
	\end{algorithm}

Although Algorithm \ref{alg:mod_policy_iter} is much more efficient compared to the generic policy iteration, one needs to evaluate $C_{\bs{h}^{(i)}}(\bs{b},s)$'s --- defined in \eqref{eq:defn_c} --- for both policy evaluation and update stages. Observe that calculating one of these quantities takes exactly $|\cV^{\leq K}| = O(|\cV^{K}|)$ iterations. We can allocate some memory to store these quantities and reduce the time complexity. First, let $\parentone(\bs{b})$ be the longest ancestor of $\bs{b}$ that is in $\cB_1$. Then, $\parentone(\bs{b}) = \parentone(\parent(\bs{b}))$ and $\parentone(\bs{b}') = \bs{b}'$ for $\bs{b}' \in \cB_1$. Moreover, let the cost of an edge between $\bs{b}$ and its parent be $b_1/\mu$ and let $\cost^{(i)}(\bs{b})$ denote the cost of going from $\bs{b}$ to $\parentone(\bs{b})$. Obviously, $\cost^{(i)}(\bs{b}') = 0$ for $\bs{b}' \in \cB_1 \cup \{\delta\}$.

Now, observe that the linear system of equations in \eqref{eqn:linear_solve_K} can be written as
\begin{align}\label{eq:linear_system_parent}
h^{(i)}(\bs{b}') + \lambda &= \eta(l(\bs{b}')-1) + E[h^{(i)}(\parent(\bs{b}')\|\bs{V}^Z)]\notag\\
&=\eta(l(\bs{b}')-1) + E\bigg[\cost^{(i)}(\parent(\bs{b}')\|\bs{V}^Z)+ h(\parentone(\parent(\bs{b}')\|\bs{V}^Z))\bigg].
\end{align}

The policy iteration algorithm starts with $s^{(0)}(\bs{b}) = l(\bs{b})$. Therefore, $h^{(0)}(\bs{b}) = \cost^{(0)}(\bs{b}) = \sum_{i< l(\bs{b})} b_i$ and one can also set $\parentone(\bs{b}) = \vmin$. We will initialize the procedure accordingly so that the policy evaluation step makes use of these quantities at the first iteration. We aim to update $\cost^{(i)}(\bs{b})$ and $\parentone(\bs{b})$ in the policy update step. To this end, we need the temporary variable 
\begin{equation}
\temp^{(i)}(\bs{b}) = \min_{s\leq l(\bs{b})} C_{\bs{h}^{(i)}}(\bs{b},s).
\end{equation}
Then the condition for the policy update becomes $C_{\bs{h}^{(i)}}(\bs{b},1) < \temp^{(i)}(\parent(\bs{b})) + b_1/\mu$; and $\temp^{(i)}(\bs{b}) $ will be updated accordingly. Note that for $b \in \cV$, $\temp^{(i)}(b) = \lambda^{(i)}$. The modified version is given in Algorithm \ref{alg:mod_policy_iter2}.

\begin{algorithm}\label{alg:mod_policy_iter2}
	
	\SetKwProg{Init}{Initialize}{}{}
	\SetKwRepeat{Do}{do}{while}
	
	\caption{Efficient Policy Iteration v2($\eta$)}		
	\KwIn{$\eta$}
	\KwOut{$J^*(\eta)$}
	\Init{}{
		$K \gets K(\eta)$\;
		$s^{(0)}(\bs{b}) \gets l(\bs{b})$\;
		$\cB_1 \gets \{v_{\min}\}$\;
		Set $\delta$ as the root of $\cT$ and add the children $\{b\|\bs{b}\}_{b \in \cV}$ for every buffer $\bs{b} \in \cT$ such that $l(\bs{b}) < K$\;
		For all $\bs{b}$, $\cost(\bs{b}) \gets \sum_{i< l(\bs{b})} b_i$ and $\parentone^{(0)}(\bs{b}) \gets\vmin$\;
		$i \gets 0$\;
	}
	\Repeat{$\bs{s}^{(i+1)} = \bs{s}^{(i)}$}{
		\tcc{Policy evaluation}
		Find $\lambda^{(i)}$, $h^{(i)}(\bs{b}')$, $\bs{b}' \in \cB_1$ by solving \eqref{eq:linear_system_parent}\;
		\For{$\bs{b} \in \cT\setminus \cB_1$ with $l(\bs{b}) > 1$}{
			$h^{(i)}(\bs{b}) \gets h^{(i)}(\parent(\bs{b})) + b_1/\mu$ if $\bs{b} \notin \cB_1$\;
		}
		\tcc{Policy update}
		$\cB_1 \gets \{v_{\min}\}$\;
		\For{$b \in \{\vmin,\ldots,v_{\max}\}$}{
			$\temp^{(i)}(b) \gets \lambda^{(i)}$.
		}
		\For{$\bs{b} \in \cT$ with $l(\bs{b}) > 1$}{
			\If{$l(\bs{b}) < K_{b_1}(\eta)$ and $C_{\bs{h}^{(i)}}(\bs{b},1) < \temp^{(i)}(\parent(\bs{b})) + b_1/\mu$}{
				$s^{(i+1)}(\bs{b}) \gets 1$\;
				Add $\bs{b}$ to $\cB_1$\;
				$\temp^{(i)}(\bs{b}) \gets C_{\bs{h}^{(i)}}(\bs{b},1)$\;
				$\cost^{(i)}(\bs{b}) \gets 0$\;
				$\parentone^{(i)}(\bs{b}) \gets \bs{b}$\;
			}
			\Else{$s^{(i+1)}(\bs{b}) = s^{(i+1)}(\parent(\bs{b}))+1$\;
				$\temp^{(i)}(\bs{b}^{(i)}) \gets \temp^{(i)}(\parent(\bs{b})) + b_1/\mu$\;
				$\cost^{(i)}(\bs{b}) \gets \cost^{(i)}(\parent(\bs{b})) + b_1/\mu $\;
				$\parentone^{(i)}(\bs{b}) \gets \parentone^{(i)}(\parent(\bs{b}))$\;
			}
		}
		$i \gets i+1$.}
	\Return $\lambda^{(i)}$
\end{algorithm}

Recall that the number of elements in $\cT$ is $O(|\cV|^{K})$. The complexity of a single iteration in Algorithm \ref{alg:mod_policy_iter2} is found as follows:
\begin{itemize}
	\item[(i)] In the policy evaluation, the equation system \eqref{eq:linear_system_parent} is constructed in $O(|\cB_1||\cV|^{K})$ time, and solved in $O(|\cB_1|^3)$.
	\item[(ii)]The policy update runs over all states. Furthermore, for each state $\bs{b}$, calculation of $C_{\bs{h}^{(i)}}(\bs{b},1) $ also requires iterations over all states. Hence, it has $O(|\cV|^{2K})$ time complexity.
\end{itemize}
As we stated before, $|\cB_1|$ is usually very small compared to $|\cV|^K$. The bottleneck then seems to be the policy update stage, which requires $O(|\cV|^{2K})$ steps. However, the policy update stage can be further improved such that the complexity decreases to $O(K|\cV|^{K})$. This modification is given in Appendix \ref{app:policy_improvement}.

Now that we have an efficient algorithm yielding $J^*(\eta)$, and with help of Corollary 1, we should be able to find the boundary curve by varying $\eta$. One may notice that initializing the trees and policies for every $\eta$ can be avoided with a minor modification. The idea is as follows: Choose a decreasing sequence $\eta_1 > \eta_2 > \ldots > \eta_n$ with $\eta_1 = \eta_{\max} = \tfrac 1 \mu (v_{\max}-\vmin)$. Note that $K(\eta_{1}) = 1$. Run the algorithm in the order of $\eta_m$'s and when $K(\eta_{m+1}) > K(\eta_{m})$, append new states to the tree $\cT$. At $(m+1)^\text{th}$ run, one can also start with the optimal policy found for $\eta_{m}$. If $\eta_1,\ldots,\eta_n$ are chosen densely, the boundary curve can be well-approximated.

%	One could also define the state space as $\cV^{\leq K}$ in the beginning of Algorithm \ref{alg:mod_policy_iter}. Although the complexity stays asymptotically same in that case, starting with a smaller state space and gradually expanding is pratically more efficient in terms of time and memory.
%	
%	\begin{theorem} The algorithm finds the optimal curve
%		\begin{equation}\label{eq:convex_conjg}
%		D(\Delta_e) = \sup_{\eta > 0}  J^*(\eta)-\eta \Delta_e.
%		\end{equation}
%	\end{theorem}
%	\begin{proof}
%		Since at most $2^{K+1}$ states are added the policy iteration converges in a finite time, and the algorithm yields an approximation to $J^*(\eta)$. Since the state space is bounded, the condition  $\sup_jE[(T_j-S_j)^2] < \infty$ from \thmref{thm:simple_age_D} is satisfied. Furthermore, we obtain a stationary policy, implying that $J^*(\eta) = D + \eta \Delta_e$. Finally, one takes the convex conjugate $D(\Delta_e) = \sup_{\eta > 0}  J^*(\eta)-\eta \Delta_e$ to obtain the boundary curve for the feasible $(\Delta_e,D)$ pairs. By choosing $\eta_{\min}$  and $\epsilon$ small enough, the curve can be approximated arbitrarily closely.
%	\end{proof}
	However, finding the boundary region everywhere would be optimistic. Theorem \ref{thm:suff_buffer_size} implies that the necessary buffer size scales with $\frac 1 \eta$. This suggests that even though the algorithm gives the almost exact curve, it is impractical to do so. To overcome this difficulty, one may rely on approximate dynamic programming algorithms; or resort to Monte Carlo estimations for the policy evaluation \cite{Bertsekas}.
	
	%	With an aim to obtain a looser achievability curve in a timely manner, we impose further limitations on the sender for the reduction of the alphabet size:
	%	\begin{itemize}
	%		\item Take a set $\cI \subset \cX$. The sender should only be able to distinguish if the data content $x$ belongs to $\cI$ or not. These packets will be called as 'important' packets. They might signify some kind of risk or importance depending on the purpose of the network, e.g. if there is a possibility of a forest fire, a major stock crash, etc. Also note that this kind of reduction is equivalent to a quantization scheme and may be extended to continuous packet contents.
	%		\item Define $v_{\max} := E[d(X,?)\indic\{X \in \cI\}]$ and $v_{\min} := E[d(X,?)\indic\{X \in \cI^C\}]$. Again, without loss of generality, one can assume $v_{\min} = 1$ or $0$. Hence, $\cV$ now only consists of $v_{\min}$ and $v_{\max}$.
	%	\end{itemize}
	
	Note that any straight line $D + \eta\Delta_e = J^*(\eta)$ in the $(\Delta_e,D)$ plane is a lower bound to the feasible region. Hence, any family of straight lines obtained in such manner gives a lower bound in general --- and if we were able to run the algorithm for every $\eta > 0$, this would give the exact boundary curve. As before, let $\eta_1 > \eta_2 > \ldots >\eta_n$ be a densely chosen sequence for which the algorithm is run. Let $\Delta_e^{(n-1,n)}$ be the abscissa of the point where the last two lines $D + \eta_{n-1}\Delta_e = J^*(\eta_{n-1})$ and $D + \eta_n\Delta_e = J^*(\eta_n)$ intersect. Then, one can see that the supremum of the straight lines obtained for $\eta_1,\ldots,\eta_n$ approximately gives the tradeoff curve for $\Delta_e \leq \Delta_e^{(n-1,n)}$, while it gives a lower bound for $\Delta_e > \Delta_e^{(n-1,n)}$. This is because all intersection points that lie on the supremum, and the line segments connecting them are achievable. This straight-line converse bound is referred as `PI converse' in the numerical examples.
	
	We end this section with some numerical examples. We have calculated the optimal policies with Algorithm \ref{alg:mod_policy_iter2}, and unfortunately we have not observed any simple structure for optimal policies for $|\cV| = 2$. We also evaluated some simple policies described below and compared their performances with the family of straight lines generated by Algorithm \ref{alg:mod_policy_iter2}, referred as `PI' in Figures \ref{fig:geom1} and \ref{fig:geom2}. These simple policies are:
	\begin{itemize}
		\item[(S1)] Send the oldest important data within a maximum buffer size $K$.
		\item[(S2)] Send the newest important data within a maximum buffer size $K$.
		\item[(S3)] Send the newest important data that has arrived more than $K$ slots ago. If there is no such data, send the oldest important one.
	\end{itemize}
	
	Each strategy above induces a finite-state Markov chain. Moreover, when $Z$ is geometrically distributed, all the Markov chains induced by these strategies have closed-form stationary distributions. $\Delta_e$ and $D$ pertaining to these strategies will accordingly have closed-form expressions. We provide these expressions in Appendix \ref{app:simple}.
	
	To compare these strategies, we also give a simple converse bound and observe their approach towards this bound for large $\Delta_e$.
	
	\begin{lemma}\label{lem:simple_cnv_1} Suppose $V = v_{i}$ with probability $\alpha_i$. Let $j^*$ be the maximum index such that $\sum_{i=j^*}^{|\cV|}\alpha_i \geq \frac 1 \mu$. Then for any $\Delta_e^{(\bs{S})}$, $D^{(\bs{S})} \geq D_{\min} =  \sum_{i = 1}^{j^*-1}\alpha_i v_i + \big(\sum_{i=j^*}^{|\cV|}\alpha_i - \frac 1 \mu\big) v_{j^*}$.
	\end{lemma}
	\begin{proof}[Proof Sketch]
		The sender can send at most $\frac 1 \mu$ fraction of the data. We then optimize over its selection of data to obtain the result.
	\end{proof}

In the first numerical example, provided in Figure \ref{fig:geom1}, $\Pr(Z=1) = 0.2$ and $\cV = \{1,20\}$ with $\Pr(V=1) = 0.7$. The blue family of straight lines correspond to the lines $D + \eta\Delta_e = J^*(\eta)$, obtained for different $\eta$ values. We could calculate until $\eta = \frac{v_{\max}-v_{\min}}{17\mu}$, which indicates that we used a maximum buffer size of 17. The region lying under this family of lines is unachievable, and the supremum of this family gives the boundary of the feasible region until the red solid line, which is the straight-line converse bound described above. The curves corresponding to strategies S1, S2 and S3 are red dashed, green dashed and cyan dotted curves, and plotted for $K \leq 20$. The simple lower bound $D_{\min} = 2.7$ is drawn as the solid black line. One can see that S2 nearly coincides with PI. Note that we observe an asymptotic behavior as the sender will never be able to allocate all of its resources to send all of the important packets.
	
	\begin{figure}[h!]
		\centering
		\includegraphics[scale = 0.65]{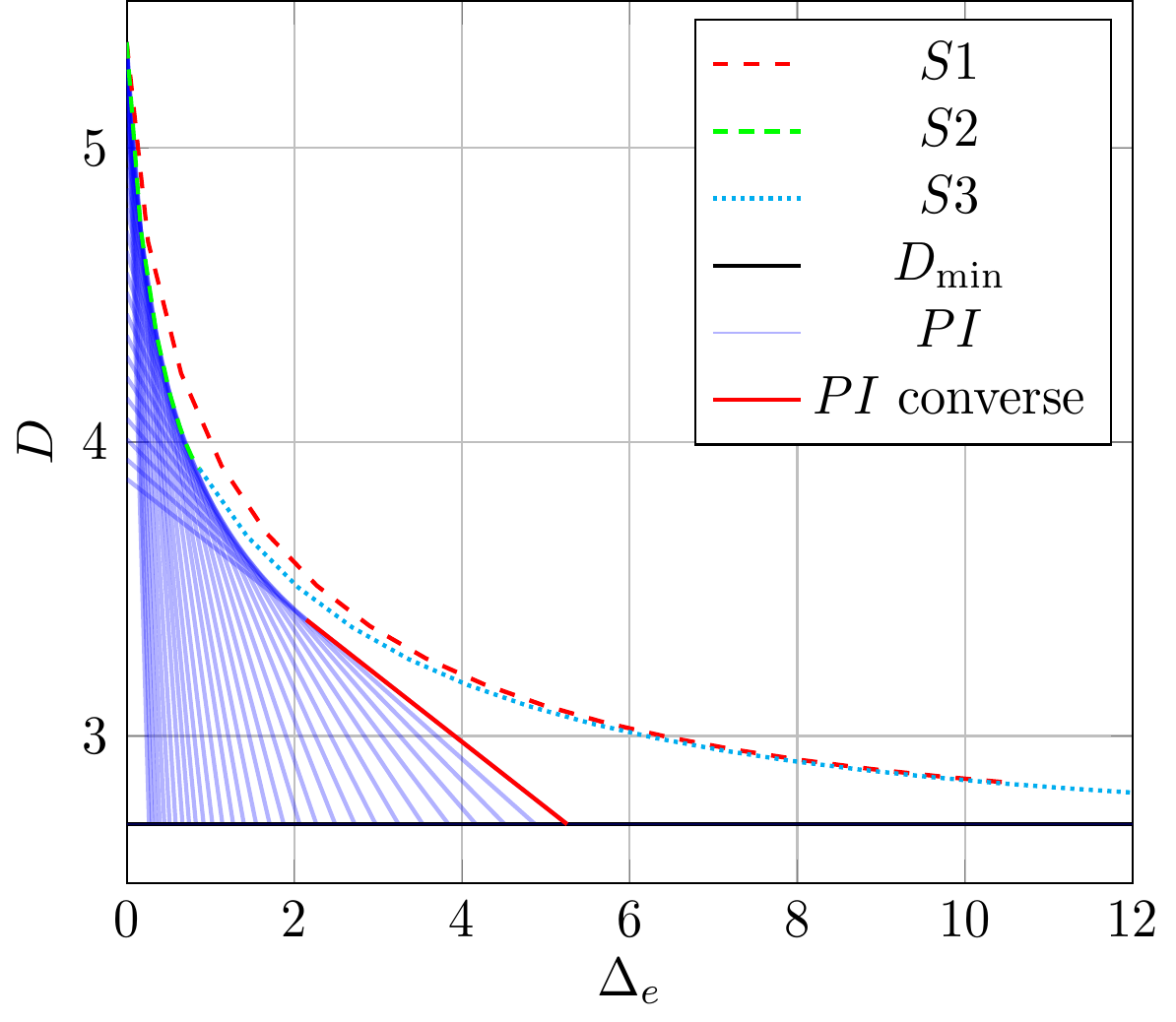}\vspace*{-5pt}\caption{Comparison of the strategies for $\cV = \{1,20\}$ and $\Pr(V = 1) = 0.7$. $Z$ is taken as a Geometric random variable with success probability $0.2$. (S2 almost coincides with PI)}\label{fig:geom1}
	\end{figure}

The second numerical example differs from the first one with $\Pr(Z=1) = 0.3$ and $\Pr(V=1) = 0.8$. The simple strategies calculated for $K\leq 40$ together with the policy iteration results obtained until a buffer size of 17 are plotted in Figure \ref{fig:geom2}. Here, $D_{\min} = 0.7$ could be achieved in finite-age as the sender can send all the important packets while keeping  $\sup_i E[(S_i-T_i)^2] < \infty$.

\begin{figure}[h]
	\centering
	\includegraphics[scale = 0.65]{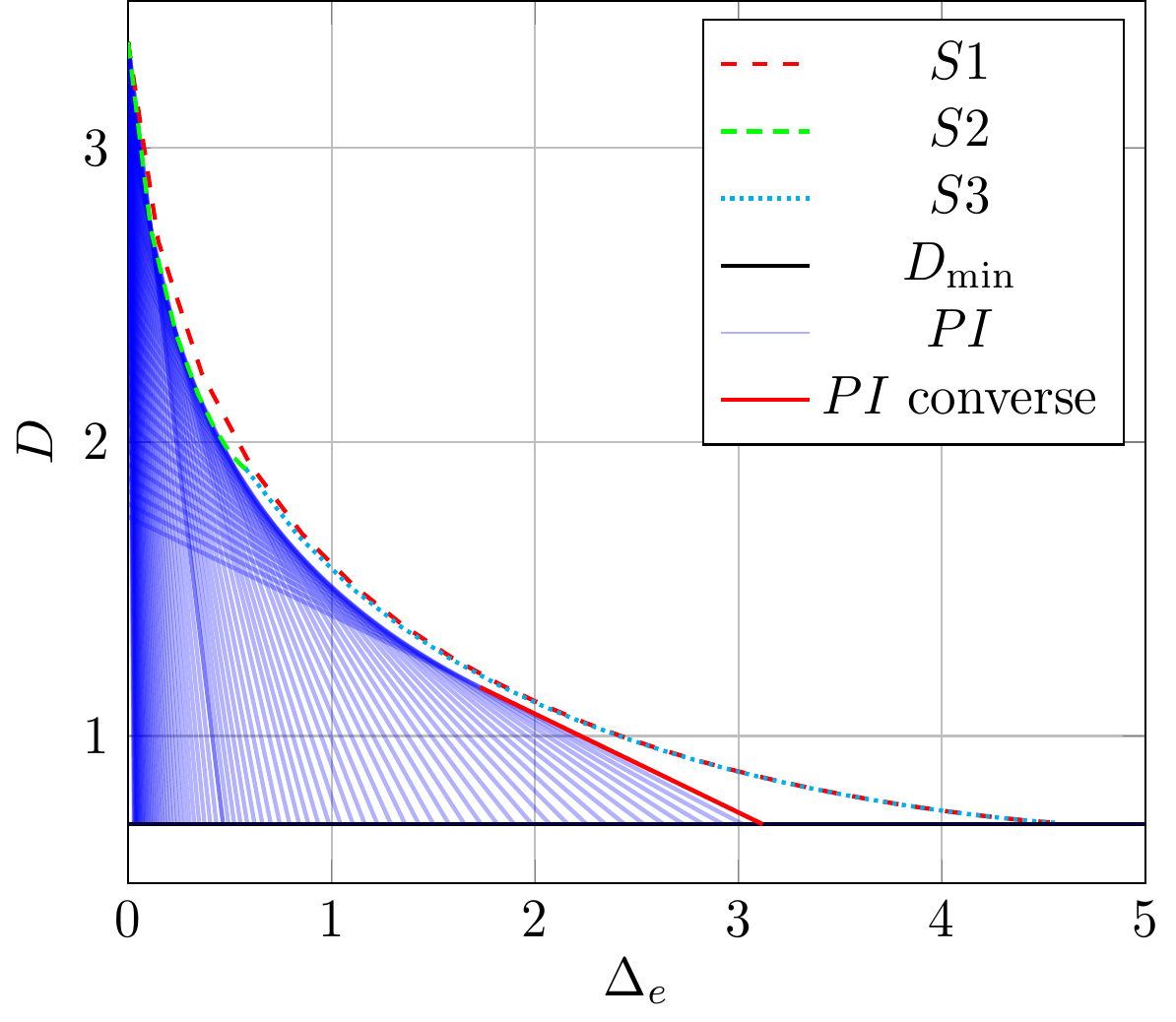}\vspace*{-5pt}\caption{Comparison of the strategies for $\cV = \{1,20\}$ and $\Pr(V = 1) = 0.8$. $Z$ is taken as a Geometric random variable with success probability $0.3$.}\label{fig:geom2}
\end{figure}

\section{Relation to an Erasure Channel with Feedback}\label{sec:erasure}

We have observed that the model discussed primarily in this work is similar to transmitting a stream of packets over an erasure channel with feedback. Recall that the setting for transmission over discrete memoryless channels with feedback requires the feedback for the data transmitted at time $t$ to be revealed just after time $t$. For an erasure channel, the knowledge of an erasure event indicator at time $t$, i.e., $\indic\{X_t \text{ is erased}\}$ is sufficient for a perfect feedback.

We note that the main difference between our model and a feedback erasure channel stems from the restriction that the feedback for data $t$ is revealed \emph{just after} time $t$. If we assume that the sender knows about a possible erasure event \emph{just before} time $t$; it may send the data $t$, or it may keep data $t$ in its buffer for a later transmission. This relaxation exactly gives the model we described with interspeaking times distributed according to a geometric distribution, i.e., $\Pr(Z_i = z) = (1-p)p^{z-1}$ where $p$ is the erasure probability for a discrete memoryless erasure channel. Also note that one can model some erasure channels with memory by varying the distribution of $Z$. 

Consider a modification of our model which requires that the constituent feedback is revealed just after the transmission. Now, we formalize the modified setting. Just before time $t$, the sender commits to a packet with timestamp $C_t \leq t$. Then, at time $t$, the committed packet $X_{C_t}$ is transmitted through the erasure channel. If the packet is erased, the sender commits to packet with timestamp $C_{t+1}$ just before time $t+1$ --- which is not necessarily equal to $C_t$ --- and the procedure is repeated until the committed data is sent. If the committed data $X_{C_t}$ is sent successfully, then $S_{i(t)+1} = C_t$ where $i(t) = \sup\{i \geq 0: T_i < t\}$; and the past of $S_{i(t)+1}$ contributes to the distortion and cannot be modified later. The age and distortion metrics are defined similar to the ones in Section \ref{sec:problem_definition}.

Observe that for any sequence of commitments $\{C_t\}_{t > 0} =: \bs{C}$, there must exist a selection procedure $\bs{S}$ in the original problem such that 
\begin{equation}\label{eq:feedback_bad}
	D^{(\bs{C})} + \eta\Delta^{(\bs{C})} \geq D^{(\bs{S})} + \eta\Delta^{(\bs{S})}.
\end{equation}

This is because in the original problem, selections are made with more information --- the erasure event is known beforehand. Conversely, for any \emph{stationary} policy $\bs{S}$ in the original problem, there exists a $\bs{C}$ in the latter problem where the sender commits to $C_t = t + s(\bs{B}_t) - l(\bs{B}_t)$ where $\bs{B}_t := V_{S_{i(t)}+1}^t$. As a consequence, the sender with no erasure information beforehand can also attain the infimal value $J^*(\eta)$, as we know from Corollary \ref{corr:results} that the optimal policy for the original problem is stationary. Together with the inequality \eqref{eq:feedback_bad}, we conclude that for both problems the tradeoff curve is the same. In brief, whether erasures are revealed just before or after transmissions do not change the tradeoff between age and distortion.
\section{When Timestamps Become Significant}\label{sec:timestamp}

We have shown that the optimal value for an $\eta > 0$ is attained with a bounded buffer policy, say of $K$. In the model so far, packets contain the timestamps as part of their headers.  Consequently, there was no need to send additional information to the receiver to tell it which packet among the $K$ packets in the buffer is being sent. If the packets do not have headers, this additional information must be included. If $K$ is much smaller compared to $|\cX|$, this additional information is insignificant.  In this section, we treat the case of headerless packets when $K$ is comparable to $|\cX|$. We take binary $\cX = \{0,1\}$ and $\cV = \{1,v\}$ with 1 being the important packet. We study the setting described in Section \ref{sec:problem_definition} but with the difference that the sender is allowed to send $N$ bits at each speaking time. We assume that $Z$ is distributed geometrically with success probability $p$, i.e. $\Pr(Z = 1) = p$.

Consider the optimal policy to attain $J^*(\eta)$ in \eqref{eq:infimum}, which is of bounded buffer size $K(\eta)$. Recall that at each speaking time, the sender is able to send one packet. If $X$ is binary, then without any coding, the optimal policy is feasible only if $N \geq 1 + \lceil\log(K(\eta))\rceil$; otherwise it is not able to describe the timestamps of selected data, e.g., for a state $\bs{b}$ with $l(\bs{b}) = K(\eta)$, the timestamps must be of length $\lceil\log(K(\eta))\rceil$ and the remaining one bit corresponds to the data. However, it seems unreasonable to use almost all of the $N$ bits for timestamp description. Can one come up with methods that do not require explicit timing information to be sent and allocate more bits to describe the data itself? The rest of this section elaborates on some possible tradeoffs with this point of view.

\subsection{Buffer Ignorant Strategies}

Think of the following strategy: Always send the most recent $N = \log(K(\eta))$ bits. In this case, it is easy to see $\Delta_e = 0$ and $D = \mu_V(1-p)^N = \mu_V K(\eta)^{\log(1-p)} \simeq \mu_V \big(\frac {v-1}{\eta}\big)^{\log(1-p)}$, where $\mu_V := E[V] = (1-q) + vq$, and $q = \Pr(V= v)$. We know that for the optimal policies described in Section \ref{sec:age_dist_tradeoff}, $D_{\min}$ given in \lemref{lem:simple_cnv_1} yields a lower bound for distortion. For a binary $X$, and thus $V$, one can obtain 
\begin{equation}
D_{\min} = \begin{cases}1-p, & p\geq q\\
\mu_V - pv& p < q
\end{cases}.
\end{equation}
If $\eta \geq \frac 1 {v-1}\big(\frac{\mu_V}{D_{\min}}\big)^{\frac 1 {\log(1-p)}}$ then $D \leq D_{\min}$, implying that the perfect timing information strategies are beaten by the timing ignorant strategy described as sending the most recent $N$ bits. In other words, one does better by sending $N$ bits of most recent data instead of sending one bit together with its timestamp.

The above arguments motivate the following question: What are the limits of these timing ignorant strategies? Note that both the sender and receiver know the speaking times $T_i$. Suppose for a moment that the receiver knows the selection times $S_i$ as well. With this assumption, the receiver has the perfect knowledge of the buffer length at time $T_i$, which is $T_i-S_{i-1}$. Hence, if the sender bases his strategies solely on its buffer size, ignoring the buffer content, it does not have to include any timing information and is able to use all its $N$-bit budget for sending data. An example could be as follows: Suppose $N=3$. Then the sender could send $1,3,5^\text{th}$ bits whenever the buffer size is $5$. Since the receiver knows the buffer size and the sender's strategy, it will know upon its reception of $3$ bits that they correspond to $1,3,5^\text{th}$ bits. Although we have previously coined the term `timing ignorant' for such strategies, a more suitable term could be `buffer ignorant'; as the sender ignores what is in its buffer.

\begin{remark}
	One may notice that this procedure is a simple online compression algorithm for the binary source $\{X_i\}$, with the restriction that the sender can only send $N$ bits at each speaking time. Together with the unsophisticated receiver that only constructs the $N$ bits it receives at each speaking time, this scheme operates at $\Delta_e  = 0$ and $D = \mu_V(1-p)^N$.
\end{remark}

Let us study the `buffer ignorant' strategies further. Since the transmitter does not use the buffer content and has to choose $N$ bits among them, a simple choice could be in size-$N$ bit contiguous chunks. Adopting the terminology from Section \ref{sec:age_dist_tradeoff}, such strategies correspond to sending $X_{S_i-N+1}^{S_i}$ at time $T_i$ from a buffer of size $L_i := T_i-S_{i-1}$. Nothing from the past $X^{S_i}$ can be sent after time $T_i$ and therefore with similar arguments we have previously done, one can formulate the problem of finding optimal selection times as a MDP. The corresponding MDP will have the buffer lengths as its states, which implies that we encounter another countable state-space problem with its state-space being $\mathbb{Z}^+$. When the buffer length is $l$, and the selection index is $s$, the one-step cost can be written as
\begin{equation}
g(l,s) =\mu_Vp(s-N)^+ + \eta(l-s).
\end{equation}
The corresponding Bellman equation is given by
\begin{equation}\label{eqn: Bellman_buff_ignorant}
h(l) + \lambda = \min_{s \leq l} \bigg\{\mu_V p (s-N)^+ + \eta(l-s) + E[h(l-s + Z)]\bigg\}
\vrule width0pt depth15pt
\end{equation}
for $l > 1$ and with $h(1) = 0$, where the buffer of length one is chosen as the reference state. It is not difficult to see that this choice also implies $h(l) = 0$ for $l \leq N$, as the sender can immediately empty the buffer for such states.

Similar to Property \ref{prop:algorithm}, we must have either $s^{*}(l) = N$ or $s^{*}(l) = s^{*}(l-1) + 1$ for the optimal policy. Therefore as $l$ increases, the optimal policy tends to leave more bits at the end. Although this observation suggests that the optimal policy may be attained with an unbounded buffer, one can show that this is not the case. Similar coupling arguments as we did in the proof of Theorem \ref{thm:suff_buffer_size} lead to the conclusion that the optimal policy cannot leave more than $\frac{N\mu_V}{\eta}$ bits at the end. Thus, a simple policy iteration algorithm run for a sufficiently large state-space also solves the Bellman equation for the infinite-state problem.

The optimal policies may not be simple-to-describe. However, when $Z$ is geometrically distributed, the numerical simulations indicate that single-threshold policies are optimal. These policies are characterized as 
\begin{equation}\label{eq:single_thresh}
s(l) = \min\{\max\{l-\tau,N\},l\}
\end{equation}
 for some $\tau \geq 0$. In other words, the sender always keeps $\tau$ unsent bits in the buffer if possible. We do not have an analytical proof for this result, but it is not unreasonable to believe that these simple policies are optimal because of the memorylessness property of $Z$. Recall that the states of Markov chains incurred by such strategies are described with the buffer length $l$. Denote the stationary probabilities as $\pi_l$. In this special case, the stationary probabilities (and consequently the average age and distortion) incurred by such strategies have closed-form expressions.

\begin{corollary}\label{corr:closed_form}The $(\Delta_e, D)$ curve attained by single-threshold strategies has a parametric description that is available in closed-form. Let $\bar{p} := 1-p$,  $S_j^{(0)} := (1+jp)$ and $S_j^{(n)} := \sum_{k=0}^j S_k^{(n-1)}$ for $n \geq 1$. Also let $S_j^{(n)} = 0$ for $j < 0$. For $0\leq j \leq \tau-1$, the stationary probabilities are given by 
\begin{equation}
\pi_{\tau-j} = \pi_{\tau + 1}\frac{1+\sum_{k=0}^{\lceil\tau/N\rceil}(-1)^{k+1}S_{j-kN}^{(k)}p^k\bar{p}^{(k+1)(N-1)}}{\bar{p}^{j+1}}
\end{equation}
	with 
\begin{equation}
\pi_{\tau + 1} = \bigg[\sum_{j=0}^{\tau-1}\frac{1+\sum_{k=0}^{\lceil\tau/N\rceil}(-1)^{k+1}S_{j-kN}^{(k)}p^k\bar{p}^{(k+1)(N-1)}}{\bar{p}^{j+1}} + \frac 1 p\bigg]^{-1}.
\end{equation}
and $\pi_{\tau + 1 + j} = \bar{p}^j\pi_{\tau + 1}$ for $j > 0$. The $(\Delta_e,D)$ curve therefore has a parametric description $(\Delta_e(\tau),D(\tau))$ given by 
\begin{equation}
\Delta_e(\tau) = \sum_{j=1}^{\tau-1} j\pi_{N+j} + \frac{\tau\pi_{\tau+1}\bar{p}^{N-1}}{p}.
\end{equation}
\centerline{and}
\begin{equation}
D(\tau)  = \frac{\mu_V \pi_{\tau+1}\bar{p}^{N}}{p^2}.
\end{equation}
\end{corollary}
\begin{remark}
	If the bits sent are equal in terms of their importance, i.e., $v_1 = v_2$, then buffer ignorant strategies are optimal among the strategies without timestamp coding. This is because strategies as such already assume that $1$ and $0$ are of equal importance and there is no need to indicate which bit is more important.  Consequently, for the binary erasure channel, i.e., $N=1$; the only optimal buffer ignorant strategy is to send the last bit if the bits are equally important.
\end{remark}

Until now, we have only considered integer $N$. However, the erasure channel between the sender and receiver could admit $K$ inputs, where $K$ is not necessarily a power of $2$. This would imply a non-integer $N = \log K$. Studying such $N$ requires strategies with some knowledge of the buffer content and possibly requires some coding. We will elaborate on how to handle these cases in the next section. 

\subsection{Revealing Partial Buffer Content}
The previous section was devoted to buffer ignorant strategies. Now, we allow some knowledge of the buffer content to improve buffer ignorant strategies and also to cover the case of a non-integer $N$.

Let us first show the improvement by coding over the buffer ignorant strategies for an integer $N$. Take the single-threshold strategy $s(l)$ described in \eqref{eq:single_thresh} together with the threshold $\tau$. Consider the state $l > \tau + N$, where the dictated action is to keep $\tau$ bits for future and send $N$ of the remaining bits. Then $l-\tau-N$ bits will never be sent and hence the distortion penalty will be $\mu_V(l-\tau-N)$. We aim to show that the distortion penalty can be decreased with some coding while preserving or decreasing the age penalty.

To that end, consider an alternative way of describing the buffer ignorant strategy above: the sender sends the first $N$ bits of the sequence $x_{l-\tau},x_{l-\tau-1},\ldots,x_{1}$ regardless of the content. It is therefore reasonable to think that a parser with a dictionary of size $2^N$ which sends the identity of the first parsed word in the same sequence could result in an improvement. One could resort to some variable-to-fixed length source coding techniques, such as Tunstall coding. Tunstall coding is known to maximize the expected number of bits parsed among prefix-free and variable-to-fixed length dictionaries. Let $E[L_{\text{Tun}}]$ be the expected number of bits parsed. With Tunstall coding, the expected number of unsent bits will be $l-\tau-E[L_{\text{Tun}}]$ and since $E[L_{\text{Tun}}] \geq N$, the distortion cost decreases. Thus, Tunstall algorithm improves the buffer ignorant strategies. Also note that for non-integer $N$, Tunstall algorithm can be used to determine parsing methods.

\begin{figure}[h]
	\centering
\includegraphics[scale=0.8]{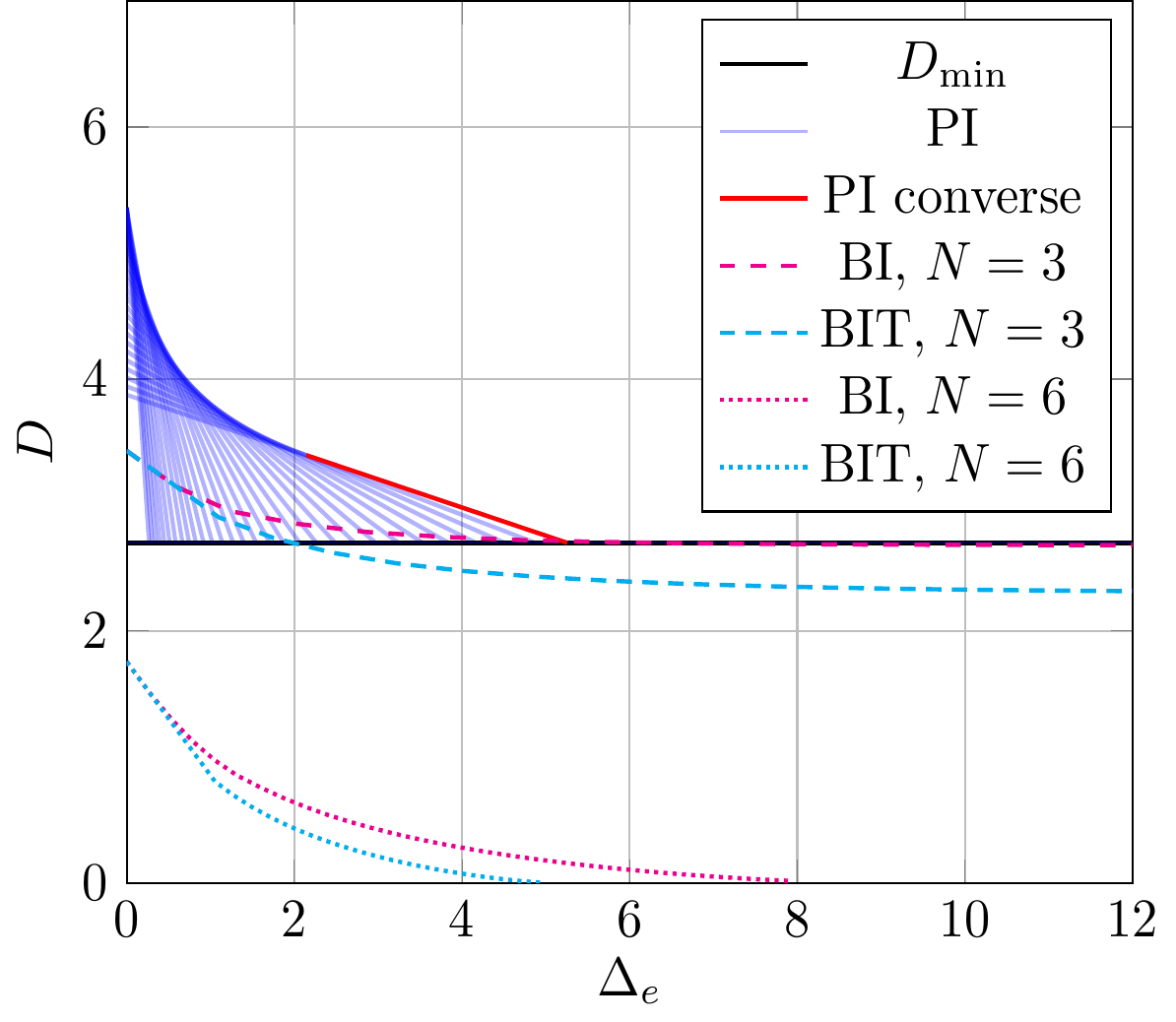}
\caption{For the same setting in Figure \ref{fig:geom1}, the curves pertaining to optimal buffer ignorant strategies (BI, colored in red) together with their improved versions with Tunstall coding (BIT, colored in cyan) are plotted for the same source in Figure \ref{fig:geom1}. The curves corresponding to $N = 3$ are dashed, and the curves corresponding to $N = 6$ are dotted.}
\label{fig:tunstall}
\end{figure}

We end this section by presenting some numerical results that illustrate the improvement with buffer ignorant strategies and Tunstall coding. We use the same source and interspeaking time distribution as in Figure 1, i.e., $\cV = \{1,20\}$ with $\Pr(V = 1) = 0.7$ and $\Pr(Z=1) = 0.2$. Recall that we were able to find the optimal policies with Algorithm \ref{alg:mod_policy_iter2} up to a buffer size of 17. This suggests that the sender describes the timing information with $\lceil\log17\rceil = 5$ bits and with an additional bit to describe the content, which implies that the left end of the PI converse line segment can be attained with sending $N = 6$ bits --- note that this converse bound is valid only for strategies of Section \ref{sec:age_dist_tradeoff}, and not for buffer ignorant strategies. As we suggested, the timing information can be sacrificed to allocate the whole budget for the data description in an attempt to improve the performance. This improvement is evident even for $N = 3$, where the optimal buffer ignorant strategy (BI, $N=3$) and its improved version with Tunstall coding (BIT, $N=3$) perform better than the optimal PI curve as seen in Figure \ref{fig:tunstall}. For $N =6$, there is drastical improvement and one can approach zero distortion with finite age --- see BI and BIT, $N=6$ in Figure \ref{fig:tunstall}.

\subsection{Discussion about other possible coding strategies}

The buffer ignorant strategies discussed in the previous section only depend on the buffer size and we have shown that these strategies can be improved by revealing some buffer content. One can argue that these coding strategies might be far from optimal, as they use partial knowledge. The sender could base its strategies on all past speaking times and the past buffer content. Calculating the tightest $(\Delta_e,D)$ curve corresponding to this broad class of strategies seems to be formidably complex. Even with the sole knowledge of the current buffer content, the problem becomes difficult. We illustrate this with an example. Suppose the current buffer content is $\bs{b} = [v,v,1,1,v,v,v,1,1,1,v,v,v,1]$ and the sender could only send $N=3$ bits. The single-threshold strategy with $\tau = 4$ will choose the index $s(\bs{b}) = 10$ and send $(1,1,1)$, which does not contain any important data. If the index $s = 13$ is chosen, $(v,v,v)$ will be sent and this might be a better strategy; however, the description of $s=13$ must be somehow included in the $N=3$ bits and all three important data might not be sent. Encoding both the data and their indices into a fixed number of bits complicates the possible actions and thus the problem appears to be very hard. We also observe an interesting tradeoff in this situation. In our MDP, sacrificing perfect state information results in a smaller policy space with policies of lower penalties, i.e., if $k$ bits are allocated for timestamp description, $N-k$ information bits can be sent and when $k$ decreases, the number of possible actions also decreases but their one step costs can possibly be smaller.

\section{Discussion}
In this work, we have studied a discrete-time model where the sender is only allowed to speak at time slots assigned by an external scheduler. In the absence of a distortion measure, it is clear that the optimal strategy is to send the freshest packet in the buffer at each speaking time as this will minimize the age. However, if this freshest packet has low importance, it may be beneficial to send a packet of higher importance instead, sacrificing freshness for lowering distortion. Hence, it is immediate that a tradeoff between the age and the distortion exists. It turns out that the optimal tradeoff can be attained with bounded buffer policies, and these policies can be found with numerical methods. Unfortunately, they turn out to be not simple-to-describe.

We observed that the usual policy iteration methods were inefficient for our specific problem, and we devised an algorithm based on appropriate data structures and problem-specific simplifications. The new algorithm performs significantly better --- the time complexity decreases from $O(|\cV|^{3K})$ to $O(K|\cV|^{K})$, where $K$ is the exact buffer size needed for attaining the optimal tradeoff. However, at the high-age regime, $K$ is large and in turn, the optimal tradeoff cannot be computed. One could try to find simple-to-describe policies that are not too far from the optimal tradeoff at this regime.

The main results of this work also apply when the process of importance levels $\{V_i\}_{i \geq 0}$ is an ergodic Markov chain --- one can verify the conditions (i) and (ii) with the same state and action spaces defined in Section \ref{sec:mdp_formulation}. Consequently, the necessary buffer size will be the same and one is able to find the optimal curve as in the i.i.d. case.

The problem we formulated in the first place turns out to be closely related to a problem of transmitting packets over an erasure channel with perfect feedback. The difference is that in our setting, erasure events are revealed before transmissions. Nevertheless, we have shown that for both problems, the optimal tradeoff is the same.

Until Section \ref{sec:timestamp}, we assumed that the timing information is contained in the header of a packet, which is much smaller in size compared to the payload. In Section \ref{sec:timestamp}, we studied the case where the packets need not contain perfect timing information, and consist of at most $N$ bits. As a consequence, if one decides to sacrifice age for lowering distortion, some additional information must be included in the $N$ bits in order to tell the receiver to which time the information bits pertain. Therefore, if the sender decides not to send the freshest data, it not only sacrifices age but may also decrease the amount of information bits to be sent. We have studied some simple policies, called buffer ignorant, where the sender ignores what is in the buffer and allocates all its $N$-bit budget for the information bits. When the timing information dominates the payload, buffer ignorant policies improve drastically over the optimal policies found in Section \ref{sec:age_dist_tradeoff}, which include the timing information. Later, we have shown that buffer ignorant policies can be further improved by revealing some buffer content and using variable-to-fixed length coding. However, it seems very challenging to find the optimal tradeoff when any strategy that sends $N$ bits at a time is allowed.

	\bibliographystyle{IEEEtran}
	\bibliography{Ref4}

\begin{thebibliography}{10}
\providecommand{\url}[1]{#1}
\csname url@samestyle\endcsname
\providecommand{\newblock}{\relax}
\providecommand{\bibinfo}[2]{#2}
\providecommand{\BIBentrySTDinterwordspacing}{\spaceskip=0pt\relax}
\providecommand{\BIBentryALTinterwordstretchfactor}{4}
\providecommand{\BIBentryALTinterwordspacing}{\spaceskip=\fontdimen2\font plus
\BIBentryALTinterwordstretchfactor\fontdimen3\font minus
  \fontdimen4\font\relax}
\providecommand{\BIBforeignlanguage}[2]{{%
\expandafter\ifx\csname l@#1\endcsname\relax
\typeout{** WARNING: IEEEtran.bst: No hyphenation pattern has been}%
\typeout{** loaded for the language `#1'. Using the pattern for}%
\typeout{** the default language instead.}%
\else
\language=\csname l@#1\endcsname
\fi
#2}}
\providecommand{\BIBdecl}{\relax}
\BIBdecl

\bibitem{InanITW}
Y.~İnan, R.~Inovan, and E.~Telatar, ``Optimal policies for age and distortion
  in a discrete-time model,'' in \emph{2021 IEEE Information Theory Workshop
  (ITW)}, 2021, pp. 1--6.

\bibitem{Kaul}
S.~Kaul, R.~Yates, and M.~Gruteser, ``Real-time status: How often should one
  update?'' in \emph{2012 Proceedings IEEE INFOCOM}, 2012, pp. 2731--2735.

\bibitem{Kaul2}
S.~Kaul, M.~Gruteser, V.~Rai, and J.~Kenney, ``Minimizing age of information in
  vehicular networks,'' in \emph{2011 8th Annual IEEE Communications Society
  Conference on Sensor, Mesh and Ad Hoc Communications and Networks}, 2011, pp.
  350--358.

\bibitem{Kaul3}
S.~Kaul, R.~Yates, and M.~Gruteser, ``On piggybacking in vehicular networks,''
  in \emph{2011 IEEE Global Telecommunications Conference - GLOBECOM 2011},
  2011, pp. 1--5.

\bibitem{8406909}
J.~P. Champati, H.~Al-Zubaidy, and J.~Gross, ``Statistical guarantee
  optimization for age of information for the d/g/1 queue,'' in \emph{IEEE
  INFOCOM 2018 - IEEE Conference on Computer Communications Workshops (INFOCOM
  WKSHPS)}, 2018, pp. 130--135.

\bibitem{GG11}
A.~Soysal and S.~Ulukus, ``Age of information in g/g/1/1 systems,'' in
  \emph{2019 53rd Asilomar Conference on Signals, Systems, and Computers},
  2019, pp. 2022--2027.

\bibitem{generalFormula}
Y.~Inoue, H.~Masuyama, T.~Takine, and T.~Tanaka, ``A general formula for the
  stationary distribution of the age of information and its application to
  single-server queues,'' \emph{IEEE Transactions on Information Theory},
  vol.~65, no.~12, pp. 8305--8324, 2019.

\bibitem{MMinfty}
C.~Kam, S.~Kompella, and A.~Ephremides, ``Age of information under random
  updates,'' in \emph{2013 IEEE International Symposium on Information Theory},
  2013, pp. 66--70.

\bibitem{multisource_Yates}
R.~D. Yates and S.~Kaul, ``Real-time status updating: Multiple sources,'' in
  \emph{2012 IEEE International Symposium on Information Theory Proceedings},
  2012, pp. 2666--2670.

\bibitem{multiplesource}
M.~Moltafet, M.~Leinonen, and M.~Codreanu, ``On the age of information in
  multi-source queueing models,'' \emph{IEEE Transactions on Communications},
  vol.~68, no.~8, pp. 5003--5017, 2020.

\bibitem{MG11_multistream}
E.~Najm and E.~Telatar, ``Status updates in a multi-stream m/g/1/1 preemptive
  queue,'' in \emph{IEEE INFOCOM 2018 - IEEE Conference on Computer
  Communications Workshops (INFOCOM WKSHPS)}, 2018, pp. 124--129.

\bibitem{8254578}
V.~Tripathi and S.~Moharir, ``Age of information in multi-source systems,'' in
  \emph{GLOBECOM 2017 - 2017 IEEE Global Communications Conference}, 2017, pp.
  1--6.

\bibitem{Updateswithqueues}
S.~K. Kaul, R.~D. Yates, and M.~Gruteser, ``Status updates through queues,'' in
  \emph{2012 46th Annual Conference on Information Sciences and Systems
  (CISS)}, 2012, pp. 1--6.

\bibitem{packetmanagement}
M.~Costa, M.~Codreanu, and A.~Ephremides, ``On the age of information in status
  update systems with packet management,'' \emph{IEEE Transactions on
  Information Theory}, vol.~62, no.~4, pp. 1897--1910, 2016.

\bibitem{9517796}
R.~D. Yates, ``The age of gossip in networks,'' in \emph{2021 IEEE
  International Symposium on Information Theory (ISIT)}, 2021, pp. 2984--2989.

\bibitem{8406966}
------, ``Age of information in a network of preemptive servers,'' in
  \emph{IEEE INFOCOM 2018 - IEEE Conference on Computer Communications
  Workshops (INFOCOM WKSHPS)}, 2018, pp. 118--123.

\bibitem{Rajai}
R.~Nasser, I.~Issa, and I.~Abou-Faycal, ``Age distribution in arbitrary
  preemptive memoryless networks,'' in \emph{2022 IEEE International Symposium
  on Information Theory (ISIT) (ISIT 2022)}, Espoo, Finland, Jun. 2022.

\bibitem{deadline}
C.~Kam, S.~Kompella, G.~D. Nguyen, J.~E. Wieselthier, and A.~Ephremides, ``Age
  of information with a packet deadline,'' in \emph{2016 IEEE International
  Symposium on Information Theory (ISIT)}, 2016, pp. 2564--2568.

\bibitem{8123937}
X.~Wu, J.~Yang, and J.~Wu, ``Optimal status update for age of information
  minimization with an energy harvesting source,'' \emph{IEEE Transactions on
  Green Communications and Networking}, vol.~2, no.~1, pp. 193--204, 2018.

\bibitem{7283009}
R.~D. Yates, ``Lazy is timely: Status updates by an energy harvesting source,''
  in \emph{2015 IEEE International Symposium on Information Theory (ISIT)},
  2015, pp. 3008--3012.

\bibitem{8437904}
S.~Farazi, A.~G. Klein, and D.~R. Brown, ``Age of information in energy
  harvesting status update systems: When to preempt in service?'' in \emph{2018
  IEEE International Symposium on Information Theory (ISIT)}, 2018, pp.
  2436--2440.

\bibitem{7308962}
B.~T. Bacinoglu, E.~T. Ceran, and E.~Uysal-Biyikoglu, ``Age of information
  under energy replenishment constraints,'' in \emph{2015 Information Theory
  and Applications Workshop (ITA)}, 2015, pp. 25--31.

\bibitem{8254156}
A.~Arafa and S.~Ulukus, ``Age-minimal transmission in energy harvesting two-hop
  networks,'' in \emph{GLOBECOM 2017 - 2017 IEEE Global Communications
  Conference}, 2017, pp. 1--6.

\bibitem{8733195}
------, ``Timely updates in energy harvesting two-hop networks: Offline and
  online policies,'' \emph{IEEE Transactions on Wireless Communications},
  vol.~18, no.~8, pp. 4017--4030, 2019.

\bibitem{8335672}
------, ``Age minimization in energy harvesting communications:
  Energy-controlled delays,'' in \emph{2017 51st Asilomar Conference on
  Signals, Systems, and Computers}, 2017, pp. 1801--1805.

\bibitem{8006703}
B.~T. Bacinoglu and E.~Uysal-Biyikoglu, ``Scheduling status updates to minimize
  age of information with an energy harvesting sensor,'' in \emph{2017 IEEE
  International Symposium on Information Theory (ISIT)}, 2017, pp. 1122--1126.

\bibitem{8437547}
S.~Feng and J.~Yang, ``Minimizing age of information for an energy harvesting
  source with updating failures,'' in \emph{2018 IEEE International Symposium
  on Information Theory (ISIT)}, 2018, pp. 2431--2435.

\bibitem{8406974}
------, ``Optimal status updating for an energy harvesting sensor with a noisy
  channel,'' in \emph{IEEE INFOCOM 2018 - IEEE Conference on Computer
  Communications Workshops (INFOCOM WKSHPS)}, 2018, pp. 348--353.

\bibitem{8422086}
A.~Arafa, J.~Yang, and S.~Ulukus, ``Age-minimal online policies for energy
  harvesting sensors with random battery recharges,'' in \emph{2018 IEEE
  International Conference on Communications (ICC)}, 2018, pp. 1--6.

\bibitem{8437573}
B.~T. Bacinoglu, Y.~Sun, E.~Uysal–Bivikoglu, and V.~Mutlu, ``Achieving the
  age-energy tradeoff with a finite-battery energy harvesting source,'' in
  \emph{2018 IEEE International Symposium on Information Theory (ISIT)}, 2018,
  pp. 876--880.

\bibitem{distGaussian}
Y.~Dong, P.~Fan, and K.~B. Letaief, ``Energy harvesting powered sensing in iot:
  Timeliness versus distortion,'' \emph{IEEE Internet of Things Journal},
  vol.~7, no.~11, pp. 10\,897--10\,911, 2020.

\bibitem{8377368}
E.~T. Ceran, D.~Gündüz, and A.~György, ``Average age of information with
  hybrid arq under a resource constraint,'' in \emph{2018 IEEE Wireless
  Communications and Networking Conference (WCNC)}, 2018, pp. 1--6.

\bibitem{wang2018skip}
B.~Wang, S.~Feng, and J.~Yang, ``To skip or to switch? minimizing age of
  information under link capacity constraint,'' in \emph{2018 IEEE 19th
  International Workshop on Signal Processing Advances in Wireless
  Communications (SPAWC)}, 2018, pp. 1--5.

\bibitem{Inan2206:Age}
Y.~Inan and E.~Telatar, ``{Age-Optimal} causal labeling of memoryless
  processes,'' in \emph{2022 IEEE International Symposium on Information Theory
  (ISIT) (ISIT 2022)}, Espoo, Finland, Jun. 2022.

\bibitem{yates2020ageSurvey}
R.~D. Yates, Y.~Sun, D.~R. {Brown III}, S.~K. Kaul, E.~Modiano, and S.~Ulukus,
  ``Age of information: An introduction and survey,'' 2020.

\bibitem{tutorial}
A.~Kosta, N.~Pappas, and V.~Angelakis, \emph{Age of Information: A New Concept,
  Metric, and Tool}, 2017.

\bibitem{AoIerror}
K.~Chen and L.~Huang, ``Age-of-information in the presence of error,'' in
  \emph{2016 IEEE International Symposium on Information Theory (ISIT)}, 2016,
  pp. 2579--2583.

\bibitem{8761668}
S.~Feng and J.~Yang, ``Age-optimal transmission of rateless codes in an erasure
  channel,'' in \emph{ICC 2019 - 2019 IEEE International Conference on
  Communications (ICC)}, 2019, pp. 1--6.

\bibitem{8362277}
A.~Baknina and S.~Ulukus, ``Coded status updates in an energy harvesting
  erasure channel,'' in \emph{2018 52nd Annual Conference on Information
  Sciences and Systems (CISS)}, 2018, pp. 1--6.

\bibitem{9440981}
S.~Feng and J.~Yang, ``Age of information minimization for an energy harvesting
  source with updating erasures: Without and with feedback,'' \emph{IEEE
  Transactions on Communications}, vol.~69, no.~8, pp. 5091--5105, 2021.

\bibitem{8849636}
A.~Arafa, J.~Yang, S.~Ulukus, and H.~V. Poor, ``Using erasure feedback for
  online timely updating with an energy harvesting sensor,'' in \emph{2019 IEEE
  International Symposium on Information Theory (ISIT)}, 2019, pp. 607--611.

\bibitem{7925903}
P.~Parag, A.~Taghavi, and J.-F. Chamberland, ``On real-time status updates over
  symbol erasure channels,'' in \emph{2017 IEEE Wireless Communications and
  Networking Conference (WCNC)}, 2017, pp. 1--6.

\bibitem{8006541}
R.~D. Yates, E.~Najm, E.~Soljanin, and J.~Zhong, ``Timely updates over an
  erasure channel,'' in \emph{2017 IEEE International Symposium on Information
  Theory (ISIT)}, 2017, pp. 316--320.

\bibitem{8445909}
H.~Sac, T.~Bacinoglu, E.~Uysal-Biyikoglu, and G.~Durisi, ``Age-optimal channel
  coding blocklength for an m/g/1 queue with harq,'' in \emph{2018 IEEE 19th
  International Workshop on Signal Processing Advances in Wireless
  Communications (SPAWC)}, 2018, pp. 1--5.

\bibitem{8006504}
E.~Najm, R.~Yates, and E.~Soljanin, ``Status updates through m/g/1/1 queues
  with harq,'' in \emph{2017 IEEE International Symposium on Information Theory
  (ISIT)}, 2017, pp. 131--135.

\bibitem{ElieErasure}
E.~Najm, E.~Telatar, and R.~Nasser, ``Optimal age over erasure channels,''
  \emph{IEEE Transactions on Information Theory}, pp. 1--1, 2022.

\bibitem{9029447}
T.~Soleymani, J.~S. Baras, and K.~H. Johansson, ``Stochastic control with stale
  information–part i: Fully observable systems,'' in \emph{2019 IEEE 58th
  Conference on Decision and Control (CDC)}, 2019, pp. 4178--4182.

\bibitem{8814627}
A.~Mitra, J.~A. Richards, S.~Bagchi, and S.~Sundaram, ``Finite-time distributed
  state estimation over time-varying graphs: Exploiting the
  age-of-information,'' in \emph{2019 American Control Conference (ACC)}, 2019,
  pp. 4006--4011.

\bibitem{Wiener_estimate}
Y.~Sun, Y.~Polyanskiy, and E.~Uysal, ``Sampling of the wiener process for
  remote estimation over a channel with random delay,'' \emph{IEEE Transactions
  on Information Theory}, vol.~66, no.~2, pp. 1118--1135, 2020.

\bibitem{OU_estimate}
T.~Z. Ornee and Y.~Sun, ``Sampling and remote estimation for the
  ornstein-uhlenbeck process through queues: Age of information and beyond,''
  \emph{IEEE/ACM Transactions on Networking}, vol.~29, no.~5, pp. 1962--1975,
  2021.

\bibitem{AoI_nonlinear}
A.~Kosta, N.~Pappas, A.~Ephremides, and V.~Angelakis, ``Age and value of
  information: Non-linear age case,'' in \emph{2017 IEEE International
  Symposium on Information Theory (ISIT)}, 2017, pp. 326--330.

\bibitem{Cost_of_delay}
------, ``The cost of delay in status updates and their value: Non-linear
  ageing,'' \emph{IEEE Transactions on Communications}, vol.~68, no.~8, pp.
  4905--4918, 2020.

\bibitem{8437591}
S.~K. Kaul and R.~D. Yates, ``Age of information: Updates with priority,'' in
  \emph{2018 IEEE International Symposium on Information Theory (ISIT)}, 2018,
  pp. 2644--2648.

\bibitem{8886357}
E.~Najm, R.~Nasser, and E.~Telatar, ``Content based status updates,''
  \emph{IEEE Transactions on Information Theory}, vol.~66, no.~6, pp.
  3846--3863, 2020.

\bibitem{8988940}
M.~Bastopcu and S.~Ulukus, ``Age of information for updates with distortion,''
  in \emph{2019 IEEE Information Theory Workshop (ITW)}, 2019, pp. 1--5.

\bibitem{Update_or_wait}
Y.~Sun, E.~Uysal-Biyikoglu, R.~D. Yates, C.~E. Koksal, and N.~B. Shroff,
  ``Update or wait: How to keep your data fresh,'' \emph{IEEE Transactions on
  Information Theory}, vol.~63, no.~11, pp. 7492--7508, 2017.

\bibitem{Bertsekas}
D.~P. Bertsekas, \emph{Dynamic Programming and Optimal Control, Vol. II},
  3rd~ed.\hskip 1em plus 0.5em minus 0.4em\relax Athena Scientific, 2007.

\bibitem{Ross}
S.~Ross, \emph{Stochastic processes}, ser. Wiley series in probability and
  statistics: Probability and statistics.\hskip 1em plus 0.5em minus
  0.4em\relax Wiley, 1996.

\bibitem{Martingales}
D.~Williams, \emph{Probability with Martingales}, ser. Cambridge mathematical
  textbooks.\hskip 1em plus 0.5em minus 0.4em\relax Cambridge University Press,
  1991.

\bibitem{Resnick}
S.~Resnick, \emph{A Probability Path}, ser. Modern Birkh{\"a}user
  Classics.\hskip 1em plus 0.5em minus 0.4em\relax Birkh{\"a}user Boston, 2003.

\end{thebibliography}
	
		\appendix
		
		\subsection{Proofs of Theorems \ref{thm:simple_age_D} and \ref{thm:reverse_fatou}}\label{app:thm1,2}
		Define $W_j := T_j - S_j$ for the rest of the proof and denote `almost surely' by \emph{a.s.}. 
		
		\subsubsection{Proof of Theorem 1}
		We first prove a convergence result in Lemma \ref{lemma:martingale} below, from which the equation \eqref{eq:simple_age}  follows as a corollary.
			\begin{lemma}\label{lemma:martingale}$\frac 1 i \sum_{j=1}^i  W_j(Z_{j+1}-\mu) \to 0$ a.s. if $\sup_j E[ W_j^2]< \infty$.
		\end{lemma}
		
		\begin{proof}
			 We use the result that if $\sum_{i} b_i/i$ converges, then $\frac 1 {i} \sum_{j \leq i} b_j \to 0$ --- known as Kronecker's Lemma \cite{Martingales}. Therefore, it is sufficient to show $\sum_i  W_i(Z_{i+1}-\mu)/i$ converges \emph{a.s.}. We now show that $M_n := \sum_{i=2}^n  W_{i-1}(Z_{i}- \mu)/{(i-1)}$, $M_1 := 0$ is a martingale with respect to the filtration $\cF_n := \sigma(Z_1,\ldots,Z_n,\bs{X}_1^{T_n})$. 
			
			Observe that $E[M_n|\cF_{n-1}] = M_{n-1} + E[ W_{n-1}(Z_{n}-\mu)|\cF_{n-1}]/(n-1) = M_{n-1}$ as $ W_{n-1}$ is $\cF_{n-1}$-measurable and $Z_{n}$ is independent of $\cF_{n-1}$ with $E[Z_{n}] =  \mu$. Since $M_n$ consists of uncorrelated increments, one can write
			\begin{equation}
			E[M_n^2] = \sum_{i=2}^n \frac{E[ W_{i-1}^2]\text{Var}(Z)}{(i-1)^2}.
			\end{equation}			
			Note that we assumed $E[Z^2] < \infty$, hence $\text{Var}(Z) < \infty$. Moreover, since $\sup_j E[ W_j^2]< \infty$, $\sum_{i} \frac{E[W_{i}^2]}{i^2} < \infty$. As a consequence, $\sup_n E[M_n^2] < \infty$ and the \emph{a.s.} convergence of $M_n$ follows from the martingale convergence theorem.
		\end{proof}
		
		Now we prove \eqref{eq:simple_D} and complete the proof of Theorem \ref{thm:simple_age_D}. Define
		\begin{equation}Y_j = 
		\begin{cases}X_j, &  j \in \bs{S}\\
		\?, &\text{else.}
		\end{cases}
		\end{equation}
		Given $t$, define $i = i(t) := \sup\{j \geq 0 : T_j \leq t\}$. Observe that for $j \leq i$, $Y_j(t) = Y_j$. Thus,
			\begin{equation}
			D_t^{(\bs{S})} = \frac 1 t \sum_{j\leq S_i} d(X_j,Y_j) + \frac 1 t \sum_{j = S_i+1}^t d(X_j,Y_j(t)).
			\end{equation}
			 Upper bound $D_t^{(\bs{S})}$ as
			\begin{equation}
			\begin{split}
			D_t^{(\bs{S})} &\leq \frac 1 t \sum_{j\leq S_i} d(X_j,Y_j) + \frac 1 t (T_{i+1}-S_i) v_{\max}\\
			&\leq \frac 1 {T_i} \sum_{j\leq S_i} d(X_j,Y_j) + \frac 1 t (Z_{i+1} + W_i) v_{\max}.
			\end{split}    
			\end{equation}
			Since $\sup_j E[W_j^2] < \infty$, $W_i$ is \emph{a.s.} finite for all $i$ and hence $Z_{i+1} + W_i$ is \emph{a.s.} finite. Thus $\frac 1 t (Z_{i+1} + W_i) v_{\max} \to 0$ \emph{a.s.} Then, we obtain
			\begin{equation}\limsup_{t \to \infty}{D_t^{(\bs{S})}} \leq \limsup_{i \to \infty} \frac 1 {T_i} \sum_{j\leq S_i} d(X_j,Y_j).\end{equation}
			Now, we lower bound $D_t^{(\bs{S})}$ as 
			
			\begin{equation}
			\begin{split}
			D_t^{(\bs{S})} \geq \frac 1 t \sum_{j\leq S_i} d(X_j,Y_j) &\geq \frac 1 {T_{i+1}} \sum_{j\leq S_i} d(X_j,Y_j)= \frac 1 {T_{i} + Z_{i+1}} \sum_{j\leq S_i} d(X_j,Y_j)
			\end{split}
			\end{equation}
			and take $\limsup$ on both sides to obtain
			\begin{equation}
			\limsup_{t \to \infty}{D_t^{(\bs{S})}} \geq \limsup_{i \to \infty} \frac 1 {T_{i} + Z_{i+1}} \sum_{j\leq S_i} d(X_j,Y_j).
			\end{equation}
			Finally, observe that $\frac {T_i} i \to \mu$ and $\frac {Z_{i+1}} i \to 0$ \emph{a.s.} Hence,
			\begin{equation}
			\begin{split}
			\limsup_{t \to \infty}{D_t^{(\bs{S})}} & = \frac 1 \mu \limsup_{i \to \infty} \frac 1 i \sum_{j\leq S_i} d(X_j,Y_j)\\
			& = \frac 1 \mu \limsup_{i \to \infty} \frac 1 i \sum_{j=1}^i D(\bs{V}^{T_j},S_j,S_{j-1}),
			\end{split}
			\end{equation}
			which ends the proof of \thmref{thm:simple_age_D}.
	\subsubsection{Proof of Theorem \ref{thm:reverse_fatou}}
Since $\sup_j E[W_j^2] < \infty$, it follows that 
	$\sup_{i} E[(\frac 1 i \sum_{j=1}^i W_j)^2] < \infty$. Thus, the family $(\frac 1 i \sum_{j=1}^i W_j)_{i \in \mathbb{N}}$ is uniformly integrable \cite[Chapter 13]{Martingales}. One can then use the reverse Fatou's lemma for uniformly integrable families \cite{Resnick} to obtain
	\begin{equation}
	\Delta_e^{(\bs{S})} = E\bigg[\limsup_{i \to \infty} \frac 1 i \sum_{j=1}^i W_j\bigg] \geq \limsup_{i \to \infty}  E\bigg[\frac 1 i \sum_{j=1}^i W_j \bigg].
	\end{equation}
	A similar reasoning for $D^{(\bs{S})}$ follows from the fact that each $D(\bs{V}^{T_j},S_j,S_{j-1})$ is smaller than $v_{\max}(W_{j-1} + Z_j)$. Therefore, $\{D(\bs{V}^{T_i},S_i,S_{i-1})\}_{i \geq 0}$ is a uniformly integrable family, and one proceeds in a similar way as above to obtain
	\begin{equation}
	D^{(\bs{S})} \geq \limsup_{i \to \infty}  \frac 1 \mu E\bigg[\frac 1 i \sum_{j=1}^i D(\bs{V}^{T_j},S_j,S_{j-1})\bigg].
	\end{equation}
	As $\limsup_{n} a_n +\limsup_{n} b_n \geq \limsup_n (a_n+b_n)$, the inequality
	$
	J(\eta)^{(\bs{S})} \leq D^{(\bs{S})} + \eta \Delta_e^{(\bs{S})}
	$
	holds. This completes the proof.
	
	\subsection{Proof of Lemma 1}\label{app:min_importance}
		We construct $\{\tilde S_i\}$ iteratively as follows: 
		\begin{equation}
		\tilde S_i = \begin{cases}
		S_i, &i<i_0\\
		\min\{s > S_{i_0}: V_s > \vmin\}\wedge T_{i_0},& i = i_0\\
			\tilde S_{i-1} + 1, & S_i \leq \tilde S_{i-1},\ i > i_0\\
		S_i, & S_i > \tilde S_{i-1},\ i > i_0
		\end{cases}.
		\end{equation}
		Verbally, at time $i_0$, $\tilde\bS$ takes the next packet whose importance value is more than $\vmin$ (if no such packet, selects the freshest one) and then selects consecutive packets irrespective of their importance values while it cannot choose anything that $\bS$ chooses. If $\tilde \bS$ can choose the packet that $\bS$ chooses at some time instant after $i_0$, it does so forever. Since $\tilde S_j \geq S_j$ for all $j$, we have for all $i$
		\begin{equation}
		\frac 1 i \sum_{j = 1}^i (T_j - \tilde S_j) \leq 	\frac 1 i \sum_{j = 1}^i (T_j -  S_j),
		\end{equation}
		and consequently $\Delta_e^{(\tilde \bS)} \leq \Delta_e^{(\bS)}$. For $i > i_0$, if $S_i = \tilde S_i$, then 
		\begin{equation}
		\frac 1 i \sum_{j=1}^i D(\bs{V}^{T_j},\tilde S_{j-1},\tilde S_{j}) \leq 	 \frac 1 i \sum_{j=1}^i D(\bs{V}^{T_j}, S_{j-1}, S_{j}).
		\end{equation}
		This is because $V_s \leq V_{\tilde s}$ for every $s \in \{S_1,\ldots, S_{i}\}$ and $\tilde s \in \{\tilde S_1,\ldots, \tilde S_{i}\}$. If $S_i < \tilde S_i$ at $i > i_0$, then $\tilde S$ must have skipped at most $\tilde S_{i_0} - S_{i_0}$ packets with minimum importance. Then we have 
		\begin{equation}
		\begin{split}
		&\frac 1 i \sum_{j=1}^i D(\bs{V}^{T_j},\tilde S_{j-1},\tilde S_{j})\leq 	 \frac 1 i \sum_{j=1}^i D(\bs{V}^{T_j}, S_{j-1}, S_{j})
		+ \frac 1 i \vmin (\tilde S_{i_0} - S_{i_0})\leq \frac 1 i \sum_{j=1}^i D(\bs{V}^{T_j}, S_{j-1}, S_{j})
		+ \frac 1 i \vmin (T_{i_0} - S_{i_0}).
		\end{split}
		\end{equation}
		Since $T_{i_0} - S_{i_0}$ is almost surely finite, $\frac 1 i (T_{i_0} - S_{i_0}) \to 0$. Therefore
		\begin{equation}
		\limsup_{i \to \infty} \frac 1 i \sum_{j=1}^i D(\bs{V}^{T_j},\tilde S_{j-1},\tilde S_{j})
		\leq 	\limsup_{i \to \infty} \frac 1 i \sum_{j=1}^i D(\bs{V}^{T_j}, S_{j-1}, S_{j})
		\end{equation}
		and consequently $D^{(\tilde \bS)} \leq D^{(\bS)}$.

		\subsection{Proof of Property \ref{prop:opt_soln}}\label{pf:prop1}
		In the following proofs, $g(\bs{b},s)$ refers to the one step costs of the MDP and $E_{\bs{b},s}[h(\bs{B}')]$ refers to the expected relative value of the next state accessed under the action $s$, i.e., $g(\bs{b},s) = \sum_{k=1}^{s-1}b_k + \eta(l(\bs{b})-s)$ and $E_{\bs{b},s}[h(\bs{B}')] = E[h(\bs{b}_{\geq s+1} \|\bs{V}^Z)]$.
			\begin{itemize}
%				\item[(i)] Since $h(\bs{b}) = b_1+\ldots + b_{l-1}$, for $s(\bs{b}) = l(\bs{b})$, the optimal relative value should satisfy $h^*(\bs{b}) \leq h(\bs{b})$.
%				
%				$$
%				h(\bs{b}) + \lambda \leq  \sum_{j = 1}^{l-1}b_j + E[h(V^Z)]
%				$$
%				
%				$$
%				h(b) = 0
%				$$
				\item[(i)] Let $\bs{q} := \bs{b} \| \bs{b}'$. Suppose $s^*(\bs{q}) \neq l(\bs{b}) + s^*(\bs{b}')$ and $s^*(\bs{q}) > l(\bs{b})$. Then by the optimality condition, $s^*(\bs{q}) = \arg\min_{s \leq l(\bs{q})}  g(\bs{q},s) + E_{\bs{q},s}[h(\bs{B}')] =  l(\bs{b})+ \arg\min_{s \leq l(\bs{b}')}   g(\bs{b}',s) + E_{\bs{b}',s}[h(\bs{B}')] = l(\bs{b}) + s^*(\bs{b}')$, which contradicts the statement. Therefore, if $ s^*(\bs{q}) > l(\bs{b})$, then $s^*(\bs{q}) = l(\bs{b}) + s^*(\bs{b}')$.

				\item[(ii)] We use $\bs{u}$ instead of $\bs{b}'$ for notational convenience and let $\bs{q} := \bs{b}\|\bs{u}$. The second upper bound is easy to show as the policy is restricted to $s > l(\bs{b})$.
				To prove $h^*(\bs{u}) \leq h^*(\bs{q})$, we will find a lower bound for $h^*(\bs{q})-h^*(\bs{u})$ and show that it is non-negative.
				Observe that
				\begin{equation}\label{eq:acc_cost}
				\begin{split}
				h^*(\bs{q})-h^*(\bs{u})&= h^*(\bs{q}) - \min_s E_{\bs{u},s}[g(s,\bs{u})+ h^*(\bs{U}') -\lambda^*]\\
				&\geq h^*(\bs{q}) - E_{\bs{u},s}[g(s,\bs{u})+ h^*(\bs{U}') -\lambda^*]\\
				&= E_{\bs{q},s^*(\bs{q})}[g(s,\bs{q})+ h^*(\bs{Q}') ] - E_{\bs{u},s}[g(s,\bs{u})+ h^*(\bs{U}') ]
				\end{split}
				\end{equation}
				for any $s$.
				Therefore, changing the actions that govern the Markov chain $\{\bs{U}_i\}$ leads to further lower bounds.
				Now, we will use a similar coupling idea as done in the proof of Lemma \ref{lem:min_importance}. Consider the two Markov chains above where the first one starts from $\bs{q}$ and the other from $\bs{u}$. Assume the two processes are coupled with having the same future arrivals for the consequent buffer states. Moreover, the former chain is controlled with its respective optimal policy $\bs{s}^*$, whereas the latter is controlled as follows: If possible, choose the data chosen by the first process; otherwise choose the oldest data possible.  Denote this policy as $\tilde{\bs{s}}$ and let $\{\bs{Q}_i\}$ be the Markov chain pertaining to the former process. Then these two processes will follow the paths
				\begin{equation}\begin{split}
				\bs{q} = \bs{Q}_0 \stackrel{s^*(\bs{Q}_0)}{\to} \bs{Q}_1 \stackrel{s^*(\bs{Q}_1)}{\to} \bs{Q}_2 \to \ldots \to \bs{Q}_{\tau}\\
				\bs{u} =\bs{U}_0 \stackrel{\tilde s(\bs{U}_0)}{\to} \bs{U}_1 \stackrel{\tilde s(\bs{U}_1)}{\to} \bs{U}_2 \to \ldots \to \bs{U}_{\tau}
				\end{split}
				\end{equation}
				and eventually end in the same state $\bs{Q}_\tau = \bs{U}_\tau$ after a random time $\tau$. This is because all possible policies induce the same recurrent class as mentioned in Lemma \ref{lem:unichain}. Replacing $h^*(\bs{Q}') $, $h^*(\bs{U}') $ in \eqref{eq:acc_cost} several times, we obtain
				\begin{equation}\label{eq:acc_cost2}
				h^*(\bs{q})-h^*(\bs{u}) \geq E\bigg[\sum_{i = 0}^\tau g(\bs{Q}_i,s^*(\bs{Q}_i))-g(\bs{U}_i,\tilde s(\bs{U}_i))\bigg].
				\end{equation}
				Observe that the right-hand side is equal to the expectation of the difference of accumulated costs incurred by the two processes until they reach the same state. Note that $\bs{U}_i$ is a suffix of $\bs{Q}_i$ for all $i$. Therefore, the age penalty of the former process will be greater. Furthermore, the latter process chooses every data that can be chosen by the former one, except the ones in the first portion $\bs{b}$, whose miss do not contribute to the distortion penalty of the latter process. Since both accumulated age and distortion penalties of the former process $\{\bs{Q}_i\}$ cannot be smaller than those of $\{\bs{U}_i\}$, the expectation in \eqref{eq:acc_cost2} is non-negative. Hence the proof is complete.
		\end{itemize}

\subsection{Proof of Theorem \ref{thm:suff_buffer_size}}\label{app:suff_buffer_size} 

For simplicity, we prove the case with $|\cV| = 3$ and we set $v_{\min} = 1$ without loss of generality. The same proof technique may be extended to larger $|\cV|$, e.g., with induction.

Assume the policy iteration algorithm is executed with a sufficiently large buffer size $M$.
We first show that packets of importance $v_2$ are never chosen by the optimal policy if choosing them incurs an age penalty greater than $K_2 := \lceil\frac{v_2-1}{\eta \mu}\rceil$, i.e., they must not be generated more than $K_2$ time slots ago. We use the same coupling idea of two controlled Markov chains with different policies, as we have done in the proofs of Lemma \ref{lem:min_importance} and Property \ref{prop:opt_soln}(ii).

Take a state $\bs{b} \in \cV^{\leq M }$ with length greater than $K_2$, with $b_1 = v_2$ (this equality is without loss of generality) and suppose $s^*(\bs{b}) = 1$. We will show that there is a better strategy than $s^*(\bs{b}) = 1$, which contradicts the optimality of $s^*$. Consider now $\tilde{s}(\bs{b}) = \min\{k > 1: b_k > 1\}$, i.e, the index of first non-1 data (whose importance is greater than 1) in the remaining buffer. Optimality of $s^*$ requires the following inequality to hold (the argument $\bs{b}$ of $\tilde s(\bs{b})$ is omitted), otherwise the policy iteration would not have converged.
\begin{equation}\label{eqn:v2}
E\big[\eta(\tilde s - 1) + h^*(b_2^l \|\bs{V}^Z) - h^*(b_{\tilde s + 1}^l \|\bs{V}^Z)\big]  -\frac 1 \mu v_2 \leq 0
\end{equation}

Now, consider two coupled processes $\{\bs{Q}_i\}$, $\{\bs{U}_i\}$ with $\bs{Q}_0 := (b_2^l \|\bs{V}^Z)$, $\bs{U}_0:= (b_{\tilde s + 1}^l \| \bs{V}^Z)$ where the actions of $\{\bs{U}_i\}$ are modified. The modification will be the same as in the proof of Property \ref{prop:opt_soln}(ii): Choose the packet chosen by the first process if possible; otherwise choose the oldest possible. As shown in the proof of Property 1(ii), altering the policy increases the expected relative value. Consequently, we obtain a lower bound to the expectation in \eqref{eqn:v2}. Again, as every policy has the same recurrent class, the two processes coincide with probability one. When they coincide, one of the following occurs:
\begin{itemize}
	\item[(i)]  $\{\bs{Q}_i\}$ misses a non-1 data that is taken by $\{\bs{U}_i\}$. Then the accumulated cost is greater than $\frac 1 \mu v_2$.
	\item[(ii)] $\{\bs{U}_i\}$ takes a packet of importance 1 at the end. Then the accumulated penalty incurred by age will be greater than $\eta l(\bs{b})$ (using the fact that $\{\bs{U}_i\}$ remains a suffix of $\{\bs{Q}_i\}$ for all $i$) and the accumulated cost from distortion will be greater than $\frac 1 \mu$ (since $\{\bs{U}_i\}$ takes an extra packet of importance 1). The total cost will be greater than $\eta l(\bs{b}) + \frac 1 \mu \geq \eta K_2(\eta)+ \frac 1 \mu \geq \frac 1 \mu v_2$.
\end{itemize}

Therefore, we conclude that $\frac 1 \mu v_2$ is a lower bound to the expectation in \eqref{eqn:v2}. Hence the left-hand side of \eqref{eqn:v2} can never be negative and $s^*$ cannot be optimal --- if it is equal to zero, then $s^*$ and $\tilde s$ are indifferent, so one can drive the process with $\tilde s$.

Now we proceed in a similar fashion to prove the case for $v_3$. Consider a state $\bs{b}$ with $l(\bs{b}) > K_3 :=  \lceil\frac{v_3-1}{\eta \mu}\rceil$, $b_1 = v_3$ and suppose $s^*(\bs{b}) = 1$. Choose $\tilde{s}(\bs{b}) = \min\{k > 1: b_k > 1,\ l(\bs{b})-k < K_{2} \text{ if }b_k = v_2\}$, i.e., take the first non-1 data but with the constraint that if it has importance $v_2$, it must be generated within the most recent $K_2$ time slots.

Define $\{\bs{Q}_i\}$ and $\{\bs{U}_i\}$ in a similar fashion. The modification done to the actions on $\{\bs{U}_i\}$ will have a minor difference compared to the previous case. Again, if $\bs{U}_i$ cannot choose a data that is chosen by $\bs{Q}_i$, it chooses the first non-1 data but with the following extra condition: If its importance is $v_2$, take it if has been generated less than $K_2$ time slots ago; otherwise skip it and do the same for the next non-1 data. This modification ensures that $\bs{U}_i$ will choose every possible data that may be chosen by $\bs{Q}_i$. When the two processes coincide, one of the following occurs:
\begin{itemize}
	\item[(i)]  $\{\bs{Q}_i\}$ misses $v_3$ that is taken by $\{\bs{U}_i\}$. Then the accumulated cost is greater than $\frac 1 \mu v_3$.
	\item[(ii)]  $\{\bs{Q}_i\}$ misses $v_2$ that is taken by $\{\bs{U}_i\}$. Then the accumulated cost incurred by age will be greater than $\eta (l(\bs{b})-K_2)$ (as $\{\bs{U}_i\}$ has remained a suffix of $\{\bs{Q}_i\}$ and has taken this $v_2$, whose index must be greater than $l(\bs{b}) - K_2$). The accumulated cost from distortion will be greater than $\frac 1 \mu v_2$. Consequently, the total cost will be greater than $\eta (l(\bs{b})-K_2) + \frac 1 \mu v_2 \geq \frac{v_3-1}{\mu} - \frac{v_2-1}{\mu} + \frac 1 \mu v_2 = \frac 1 \mu v_3$.
	\item[(iii)] $\{\bs{U}_i\}$ takes a 1 at the end. Then the accumulated cost incurred by age will be greater than $\eta l(\bs{b})$ (again, using the fact that $\{\bs{U}_i\}$ remains a suffix of $\{\bs{Q}_i\}$) and the accumulated cost from distortion will be greater than $\frac 1 \mu$ (since $\{\bs{U}_i\}$ takes an extra 1). The total cost will be greater than $\eta l(\bs{b}) + \frac 1 \mu \geq \frac 1 \mu v_3$.
\end{itemize}

Similar to the case of $v_2$, the expectation \eqref{eqn:v2} can never be negative and hence $s^*$ cannot be optimal.

All in all, we have shown that starting the policy iteration algorithm with a buffer size $M > K_3$, the algorithm terminates with an optimal policy that does not use more than a buffer size of $K_3$. This implies that the solution of the Bellman equation lies within buffers of size at most $K_3$.

%\emph{Note:} Unlike our previous proof of \thmref{thm:suff_buffer_size} in \cite{InanITW}, where we have shown that the state $\cV^{\leq K_2}$ is invariant under policy iteration, in this proof we have shown that the optimal policy must reside in $\cV^{\leq K_3}$ no matter what $\cV^{\leq M}$, $M > K_3$ the algorithm starts from. The previous proof was only valid for $|\cV| = 2$. This updated proof not only yields the same conclusion with the previous one but also covers the case when $|\cV| > 2$.
\subsection{Policy Update Improvement}\label{app:policy_improvement}
As discussed before, the policy update stage of Algorithm \ref{alg:mod_policy_iter2} can be improved to run in $O(K|\cV|^K)$ time. Our aim is to improve the calculation of $C_{\bs{h}^{(i)}}(\bs{b},1)$ at step 16 of Algorithm \ref{alg:mod_policy_iter2}. Recall that $p_z := \Pr(Z=z)$, and $q_z := \Pr(Z\geq z)$. We omit the argument of $l(\bs{b})$ for brevity. Let us rewrite $C_{\bs{h}^{(i)}}(\bs{b},1)$ by first conditioning on $Z=z$ as
\begin{align}
C_{\bs{h}^{(i)}}(\bs{b},1) &= \eta(l -1) + \sum_{z = 1}^{\infty}p_z E[h^{(i)}(\bs{b}_{\geq 2}\| \bs{V}^z)]\notag\\
&=\eta(l -1) + \sum_{z = 1}^{K-l} p_z E[h^{(i)}(\bs{b}_{\geq 2}\| \bs{V}^z)]+  p_{K-l+1} E[h^{(i)}(\bs{b}_{\geq 2}\| \bs{V}^{K-l+1})]\label{eq:K_cost1}\\
& \phantom{=}+  \sum_{z = K-l+2}^{\infty} p_z E[h^{(i)}(\bs{b}_{\geq 2}\| \bs{V}^z)]\label{eq:K_cost2}
\end{align}
and note that for $z > K-l+1$, $h^{(i)}(\bs{b}_{\geq 2}\| \bs{V}^z) = b_2/\mu + h^{(i)}(\bs{b}_{\geq 3}\| \bs{V}^z)$ since the length of $\bs{b}_{\geq 2}\| \bs{V}^z$ exceeds $K$. In line with this result, we rewrite the summation of \eqref{eq:K_cost1} and \eqref{eq:K_cost2} as 
\begin{align}
\kappa^{(i)}(\bs{b}_{\geq 2},K) &:= p_{K-l+1} E[h^{(i)}(\bs{b}_{\geq 2}\| \bs{V}^{K-l+1})]
+ q_{K-l+2}b_2/\mu + \sum_{z = K-l+2}^{\infty} p_z E[h^{(i)}(\bs{b}_{\geq 3}\| \bs{V}^z)]\notag\\
&= p_{K-l(\bs{b}_{\geq 2})} E[h^{(i)}(\bs{b}_{\geq 2}\| \bs{V}^{K-l(\bs{b}_{\geq 2})})]
+ q_{K-l(\bs{b}_{\geq 2})+1}b_2/\mu + \sum_{z = K-l(\bs{b}_{\geq 2})+1}^{\infty} p_z E[h^{(i)}(\bs{b}_{\geq 3}\| \bs{V}^z)]
\end{align}
The key observation here is that the last term is equal to $\kappa^{(i)}(\bs{b}_{\geq 3},K) = \kappa^{(i)}(\parent(\bs{b}_{\geq 2}),K)$. Therefore, we have the recursive relation 
\begin{equation}\label{eq:K_recursive}
\kappa^{(i)}(\bs{b},K) := p_{K-l(\bs{b})} E[h^{(i)}(\bs{b}\| \bs{V}^{K-l(\bs{b})})]\\
+ q_{K-l(\bs{b})+1}b_1/\mu + \kappa^{(i)}(\parent(\bs{b}),K)
\end{equation}
with the initial condition 
\begin{equation}
\kappa^{(i)}(\delta,K) := q_KE[h^{(i)}(\bs{V}^{K})] + \frac {E[V]} \mu E[(Z-K)^+].
\end{equation}
The recursive relation \eqref{eq:K_recursive} can be implemented exactly the same as the other updates that take place in lines 17-26 of Algorithm \ref{alg:mod_policy_iter2}. Finally, we have 
\begin{equation}
C_{\bs{h}^{(i)}}(\bs{b},1) =\eta(l -1) + \sum_{z = 1}^{K-l} p_z E[h^{(i)}(\bs{b}_{\geq 2}\| \bs{V}^z)] + \kappa^{(i)}(\parent(\bs{b}),K).
\end{equation}
Knowing $\kappa^{(i)}(\parent(\bs{b}),K)$, both $C_{\bs{h}^{(i)}}(\bs{b},1) $ and $\kappa^{(i)}(\bs{b},K)$ can be calculated in $O(|\cV|^{K-l})$ steps. Consequently, the total amount of calculations done for all length-$l$ states is $O(|\cV|^{K})$, and consequently for the depth-$K$ tree $\cT$, the amount of calculations done is $O(K|\cV|^{K})$.

\subsection{Closed-Form Expressions for the Simple Strategies}\label{app:simple}
\subsubsection{S1 --- Send the oldest important data among the most recent $K$}
For a buffer $\bs{b}$, let $s$ be the index of the first important data among the $K$ most recently arrived packets. Let $a = (K \wedge l(\bs{b}))-s+1$; and if there is no important data among the $K$ most recent packets, set $a = 0$. Let $A_i$ be
the value of $a$ at time instant $i$.  It can be verified that the process $\{A_i\}$ is a Markov chain with state space $\{0,1,\ldots,K\}$. Let $p_{a,a'}$ denote the transition probability from state $a$ to $a'$; and $\pi_a$ denote the stationary probability of state $a$. Recall that $Z$ is a geometric random variable and denote $p := \Pr(Z = 1)$ and $q := \Pr(V = v_{2})$. Let $\bar x := 1-x$. 

Now, let us calculate $p_{a, a'}$.
Observe that when we are at state $a$, the first element in the buffer is selected and the remaining $a-1$ elements are untouched and need not be known. Conditioned on the next speaking time being $a'-a <  z \leq K-a$, the probability that the next state being $a'$ is equal to $\bar q^{z-(a'-a)-1} q$. This is because the first $z-(a'-a)-1$ elements must be $v_{1}$ and the next should be $v_{2}$. For $z > K-a$, since we only check the most recent $K$ data in the buffer, the first $K-a'$ must be $v_{1}$ and the next one should be $v_{2}$. Thus we obtain the probability as $\bar q^{K-a'}q$. Averaging over $z$, we have
\begin{equation}
p_{a, a'} = \sum_{z=a'-a+1}^\infty \bar p^{z-1}p \bar q^{(z-a'+a-1) \wedge (K-a')}q.
\end{equation}
With a similar reasoning, we calculate $p_{a', a}$ as 
\begin{equation}
p_{a',a} = \sum_{z=1}^\infty \bar p^{z-1}p \bar q^{(z+a'-a-1) \wedge (K-a)}q = \sum_{z=a'-a+1}^\infty \bar p^{z-a'+a-1}p \bar q^{(z-1) \wedge (K-a)}q= \bigg(\frac{\bar q }{\bar p}\bigg)^{a'-a}p_{a,a'}.
\end{equation}
It is easy to see that $p_{0,a} = p_{1,a}$ as the buffer is emptied in both cases; and $p_{a,0} = p_{a,1} \tfrac {\bar q} q$ as the only difference between the transition from state $a$ to state 0, or to state 1 is the most recent element being $v_{1}$ or $v_{2}$.

We now claim that the chain is reversible. The claim is easily verified by noting that the distribution $(\pi_a:a=0,\dots,K)$ described by
$$
\pi_0=(\bar q/q)\pi_1,\ \pi_a = (\bar q/\bar p)^{K-a}\pi_K,\ a=1,\dots,K-1
$$
\centerline{with}
\begin{equation}
\pi_K = \big(1-\tfrac{\bar q}{\bar p}\big) / \big(1-\tfrac {p}{q}\big(\tfrac{\bar q }{\bar p}\big)^K\big)
\end{equation}
satisfies the detailed balance equations $\pi_a p_{a,a'}=\pi_{a'}p_{a,a'}$.

% It is easy to see that $p_{0,a} = p_{1,a}$ as the buffer is emptied in both cases; and $p_{a,0} = p_{a,1} \tfrac {\bar q} q$ as the only difference between the transition from state $a$ to state 0, or to state 1 is the most recent element being $v_{1}$ or $v_{2}$. Then follows $\pi_0 = \pi_1 \frac {\bar q}{q}$. All of the above suggest that $\pi_{a} = \pi_{a'} \big(\frac{\bar q }{\bar p}\big)^{a'-a}$ for $1\leq a \leq a'$ and $\pi_0 = \pi_1 \frac {\bar q}{q}$: this choice of $\pi_a$'s solve the detailed balance equations for every possible $(a,a')$ pair. Hence, the chain must be reversible with the stationary distribution being determined by $\pi_a$'s. We can obtain $\pi_a$'s explicitly by finding $\pi_K$. Observe that if $q \neq p$
%\begin{equation}\label{eq:pi_K}
%\sum_{k = 0}^K \pi_k = \pi_K \bigg(\frac{1-\big(\tfrac{\bar q}{\bar p}\big)^K}{1-\tfrac{\bar q}{\bar p}} + \tfrac{\bar q } {q}\big(\tfrac{\bar q}{\bar p}\big)^{K-1}\bigg) = 1
%\end{equation}
%and after some simplification,

The average excess age is then calculated straightforwardly:
\begin{equation}
\Delta_e^{(S1)} (K) = \sum_{k = 1}^K (k-1)\pi_k \label{eq:age_s1}= \frac{K-1}{\big(1-\tfrac {p}{q}\big(\tfrac{\bar q }{\bar p}\big)^K\big)} - \frac{\tfrac{\bar q}{\bar p}\big(1-\big(\tfrac{\bar q}{\bar p}\big)^K\big)}{\big(1-\tfrac {p}{q}\big(\tfrac{\bar q }{\bar p}\big)^K\big)\big(1-\frac{\bar q}{\bar p}\big)}.
\end{equation}
For $q = p$, one takes $\frac{\bar q}{\bar p} \to 1$ in \eqref{eq:age_s1} to obtain $\Delta_e^{(S1)}  = \frac{(K-1)K}{2(K + \frac{\bar{p}}{p})}$.

The average distortion is calculated as follows: Unimportant packets are sent $\pi_0$ fraction of the speaking times, and the remaining time is allocated to the transmission of important packets. Hence, unimportant packets are sent $p\pi_0$ fraction of the time and important packets are sent  $p(1-\pi_0)$ fraction of the time. Consequently, 
\begin{equation}\label{eq:dist_s1}
D^{(S1)}(K) = (\bar q - p\pi_0)v_{1} + (q - p(1-\pi_0))v_
{2}.
\end{equation}

\subsubsection{S2 --- Send the newest important data among the most recent $K$}
Compared to the other strategies, the analysis will be relatively simpler. For a buffer $\bs{b}$, let $s$ be the index of the newest important data among the $K$ most recent and let the state be $a = l-s+1$. If there is no important data among the $K$ most recent, set $a = 0$. Observe that regardless of the value of $a$, the packets (if any) that remain the buffer after transmission are of minimum importance and will be ignored at the next speaking time. Hence, at the next speaking time, the next state $a'$ will not depend on the current state $a$, and we have an i.i.d.\ process. Let $\pi_a$ be the probability of the next state being equal to $a$. Conditioned on the next speaking time $z$, this probability is equal to $q\bar q^{a-1}\indic\{a \leq z\}$. Hence, for $0< a \leq K$
\begin{equation}
\pi_a = \sum_{z \geq a}^\infty q\bar q^{a-1}\bar p^{z-1}p = q (\bar p \bar q)^{a-1},
\end{equation}
and $\pi_0 = 1-\sum_{a =1}^K \pi_a = \frac{\bar q p + q(\bar p \bar q)^K }{1-\bar p \bar q}$. The age and distortion are calculated similar to \eqref{eq:age_s1} and \eqref{eq:dist_s1}.

\subsubsection{S3 --- Send the newest important data that has arrived more than $K$ slots ago. If there is no such data, send the oldest important one}
We set the state associated to a buffer $\bs b$ as follows: If $\bs{b}$ does not contain an important data arrived more than $K$ time slots ago, set $a = (K \wedge l(\bs{b}))-s+1$ as the state; and if there is no important data in the buffer, set $a = 0$. Note that this is exactly the same as in S1. If $\bs{b}$ does contain an important data that has arrived more than $K$ time slots ago, set $a = K+1$. Hence the state space will be $\{0,1,\ldots,K+1\}$. Similar to the analysis of S1, $\{A_i\}$ will be a Markov chain. We want to calculate the transition probabilities $p_{K+1,a}$ for $0< a < K+1$. Since at state $K+1$, an important data has arrived more than $K$ slots ago, $\bs{b}$ is of the form $\bs{b} = [v_{2},v_1,\ldots,v_1, \underbrace{\ldots}_{K}]$, where the last $K$ data need not be known. Conditioned on the speaking time $z$, the probability of ending up in state $a$ is then should be $\bar q^{z+K-a} q$. Consequently,
\begin{equation}
p_{K+1,a} = \sum_{z=1}^\infty \bar p^{z-1}p\bar q^{K-a+z}q  = \frac{pq\bar q^{K-a+1}}{1-\bar p \bar q}.
\end{equation}
Now, we aim to find $p_{a,K+1}$ for $0< a < K+1$. Since $\bs{b}$ is of the form $\bs{b} = [\ldots, v_{2}, \underbrace{\ldots}_{a-1}]$, and the last $a-1$ data need not be known, the next speaking time should be greater than $K-a+1$. Moreover, there must be at least one important data among the first $z-(K-a+1)$. Then, we obtain
\begin{equation}
p_{a,K+1} =  \sum_{z > K-a+1} \bar p^{z-1}p (1-\bar q^{z-(K-a+1)})= 
 \frac{q\bar p^{K-a+1}}{1-\bar p \bar q}.
\end{equation}

Calculation of $p_{a,a'}$, $0< a \leq a'< K+1$ is similar to the one in S1. Since $\bs{b} = [\ldots, v_{2}, \underbrace{\ldots}_{a-1}]$, the next speaking time should be greater than $a'-a$, and the first $z-(a'-a+1)$ data must be unimportant while the $z-(a'-a)^\text{th}$ data must be important. Thus,
\begin{equation}
p_{a,a'} =  \sum_{z > a'-a} \bar p^{z-1}p \bar q^{z-a'+a}q= \frac{pq\bar p^{a'-a}}{1-\bar p \bar q}.
\end{equation}
A similar analysis reveals that for $0< a' \leq a< K+1$,
\begin{equation}
p_{a,a'} =  \sum_{z = 1}^\infty \bar p^{z-1}p \bar q^{z-a'+a-1}q= \frac{pq\bar q^{a-a'}}{1-\bar p \bar q}.
\end{equation}
Finally, as we discussed in the analysis of S1, $p_{0,a} = p_{1,a}$ and $p_{a,0} = p_{a,1} \tfrac {\bar q} q$. Hence, we have found all the transition probabilities. We will verify that $\{A_i\}$ is a reversible Markov chain. The detailed balance equations for the state pairs $(a,K+1)$ with $a > 0$ require the stationary distribution $\pi$ to satisfy
\begin{equation}\label{eq: det_bal}
\pi_{K+1}p_{K+1,a} = \pi_{a}p_{a,K+1}
\end{equation}
and thus we find $\pi_a = p\big(\tfrac{\bar q}{\bar p}\big)^{K+1-a}\pi_{K+1}$ for $a>0$.  Detailed balance equations for the state pair $(0,1)$ further yield $\pi_0 = (\bar q/q)\pi_1=\tfrac{p \bar q}{q}\big(\tfrac{\bar q}{\bar p}\big)^{K}\pi_{K+1}$. One can easily verify that this choice satisfies not only \eqref{eq: det_bal} but all the detailed balance equations. Thus, we have verified the reversibility of the chain and found its stationary distribution.

Let us start the chain with the stationary distribution and let $\Delta_i := l(\bs{B}_i) - s(\bs{B}_i)$ be the instantaneous excess age at the $i^\text{th}$ speaking time. If $0\leq A_i \leq K$, then $\Delta_i = (A_i-1)\vee 0$. When $A_i = K+1$, however, $\Delta_i$ is random. It turns out that the distribution of $\Delta_i - (K-1)$ conditioned on $A_i = K+1$ is geometrically distributed. To see this, observe
\begin{equation}
\begin{split}
\Pr(\Delta_i  = K -1+ z'| A_i = K+1)&=  \frac{\Pr(\Delta_i  = K -1+ z', A_i = K+1)}{\pi_{K+1}} \\
 &= \frac 1 {\pi_{K+1}}\sum_{a = 1}^{K+1} \pi_a  \sum_{z \geq K-a+1+z'}\bar p^{z-1}p q \bar q^{z'-1}+ \frac{\pi_0}{\pi_{K+1}}\sum_{z \geq K+z'}\bar p^{z-1}p q \bar q^{z'-1} \\
 &= \frac 1 {\pi_{K+1}}\sum_{a = 1}^{K+1} \pi_a  \bar p^{K-a+z'} q \bar q^{z'-1}+ \frac{\pi_0}{\pi_{K+1}}\bar p^{K+z'-1} q \bar q^{z'-1}\\
 &= \sum_{a = 1}^{K+1} p^{\indic\{a \leq K\}}\big(\tfrac{\bar q}{\bar p}\big)^{K+1-a}  \bar p^{K-a+z'} q \bar q^{z'-1}
+ p\tfrac {\bar q} {q} (\tfrac{\bar q}{\bar p}\big)^{K}\bar p^{K+z'-1} q \bar q^{z'-1}\\
 &= (\bar p \bar q)^{z'-1}(1-\bar p \bar q).
 \end{split}
\end{equation}

Then, the average excess age is found as
\begin{equation}
\Delta_e^{(S3)}(K) = \sum_{k = 1}^K (k-1)\pi_k + \pi_{K+1} \big(\tfrac 1 {1-\bar p \bar q} + K-1\big).
\end{equation}

The average distortion is calculated the same as in \eqref{eq:dist_s1}.

\subsection{Proof of Corollary \ref{corr:closed_form} }\label{pf:closed_form}

Denote the stationary probabilities of each state $l$ as $\pi_l$ and $p_i := \Pr(Z = i) = p\bar{p}^{i-1}$ with $\bar{p} := 1-p$. Assume $\tau > N$ for the moment. The derivation for $\tau \leq N$ will easily follow.

Let $s_i^j := \sum_{l=i}^j \pi_l$ and $s_i := s_i^\infty$.
The stationary probabilities are the solution to the linear system
\begin{equation}\label{eqn:lin_sys}
\begin{split}
\pi_1 &= p_1s_1^N\\
\pi_2 &= p_2s_1^N+ p_1\pi_{N+1}\\
\ldots\\
\pi_{\tau} &= p_{\tau}s_1^N+ p_{\tau-1}\pi_{N+1} + \ldots + p_{1}\pi_{N+\tau-1}\\
\pi_{\tau+1} &= p_{\tau+1}s_1^N+ p_{\tau}\pi_{N+1} + \ldots + p_{1}s_{N+\tau}\\
\ldots\\
\pi_{\tau+j} &= p_{\tau+j}s_1^N + p_{\tau+j-1}\pi_{N+1} + \ldots + p_{j}s_{N+\tau}\\
\ldots
\end{split}
\end{equation}
Observe that 
\begin{equation}\label{eqn:stationary1}
	\pi_{j+1} = \bar{p}\pi_{j} + p\pi_{N+j},\ 1 \leq j \leq \tau-1
\end{equation}
\centerline{and}
\begin{equation}\label{eqn:stationary2}
\pi_{j+1} = \bar{p}\pi_{j},\ j \geq \tau+1.
\end{equation}
Summing up the equalities in \eqref{eqn:stationary1}, with indices up to $j+1$, we obtain
\begin{equation}
s_2^{j+1} = \bar{p}s_1^j + ps_{N+1}^{N+j}
\end{equation}
and hence
\begin{equation}
\pi_{j+1} +p s_2^j = \bar{p}\pi_1 + ps_{N+1}^{N+j}.
\end{equation}
Note that the first equation in \eqref{eqn:lin_sys} implies $\bar{p}\pi_1 = ps_2^N$ and thus we get
\begin{equation}
\pi_{j+1} +p s_2^j = ps_2^N+ ps_{N+1}^{N+j}
\end{equation}
\centerline{and}
\begin{equation}
\bar{p}\pi_{j+1} = ps_{j+2}^{N+j},\ 1 \leq j \leq \tau-1,
\end{equation}
which implies for $j = \tau-1$
\begin{equation}
\pi_{\tau} = p\sum_{k= 1}^{N+\tau-1}\pi_{\tau+k}  = \frac{\pi_{\tau+1}(1-\bar{p}^{N-1})}{\bar{p}}
\end{equation}
where the last equality follows from \eqref{eqn:stationary2}.
Now repeated application of \eqref{eqn:stationary1} gives
\begin{equation}\label{eqn:iterative_tau}
\pi_{\tau-j} = \pi_{\tau+1}\frac{1-(1+jp)\bar{p}^{N-1}}{\bar{p}^{j+1}},\ 0 \leq j \leq N-1.
\end{equation}
Our aim is now to find all stationary probabilities in terms of $\pi_{\tau+1}$. With the above, we are able to find $\pi_{\tau-j}$, $0 \leq j \leq N-1$ in terms of $\pi_{\tau+1}$. For $j > N-1$, we try to observe a pattern. First, try to calculate $\pi_{\tau-N}$ by  \eqref{eqn:stationary1}, which gives
\begin{equation}
\pi_{\tau-N} = \pi_{\tau+1}\frac{1-(1+Np)\bar{p}^{N-1} + p \bar{p}^{2(N-1)}}{\bar{p}^{N+1}}.
\end{equation}
Once more, repeated application of \eqref{eqn:stationary1} gives 
\begin{equation}
\begin{split}
\pi_{\tau-N-j}=\pi_{\tau+1}\frac{1-(1+(N+j)p)\bar{p}^{N-1}}{\bar{p}^{N+j+1}}+ \pi_{\tau+1} \frac{p(\sum_{k=0}^j(1+kp)) \bar{p}^{2(N-1)}}{\bar{p}^{N+j+1}}
\end{split}
\end{equation}
for $0 \leq j \leq N-1$. Doing the same procedure, we observe the following pattern: Let $S_j^{(0)} := (1+jp)$ and $S_j^{(n)} := \sum_{k=0}^j S_k^{(n-1)}$ for $n \geq 1$. Also let $S_j^{(n)} = 0$ for $j < 0$. 
Then,
\begin{equation}
\pi_{\tau-j} = \pi_{\tau + 1}\frac{1+\sum_{k=0}^{\lceil\tau/N\rceil}(-1)^{k+1}S_{j-kN}^{(k)}p^k\bar{p}^{(k+1)(N-1)}}{\bar{p}^{j+1}}.
\end{equation}
Finally, since the probabilities sum up to one, we have 
\begin{equation}
\begin{split}
1&= s_1^{\tau} + s_{\tau+1} = s_1^{\tau} + \frac{\pi_{\tau+1}}{p} \\
&= \pi_{\tau + 1}\sum_{j=0}^{\tau-1}\frac{1+\sum_{k=0}^{\lceil\tau/N\rceil}(-1)^{k+1}S_{j-kN}^{(k)}p^k\bar{p}^{(k+1)(N-1)}}{\bar{p}^{j+1}}+ \frac{\pi_{\tau+1}}{p}
\end{split}
\end{equation}
and therefore
\begin{equation}
\pi_{\tau + 1} = \bigg[\sum_{j=0}^{\tau-1}\frac{1+\sum_{k=0}^{\lceil\tau/N\rceil}(-1)^{k+1}S_{j-kN}^{(k)}p^k\bar{p}^{(k+1)(N-1)}}{\bar{p}^{j+1}} + \frac 1 p\bigg]^{-1}.
\end{equation}
After calculating all $\pi_j$s, it is straightforward to obtain expressions for $\Delta_e$ and $D$ as
\begin{equation}
\Delta_e = \sum_{j=1}^{\tau-1} j\pi_{N+j} + \tau \sum_{j = \tau}^\infty \pi_{N+j} = \sum_{j=1}^{\tau-1} j\pi_{N+j} + \frac{\tau\pi_{\tau+1}\bar{p}^{N-1}}{p}
\end{equation}
\centerline{and}
\begin{equation}
D = \mu_V \sum_{j = 1}^{\infty}j\pi_{\tau+N+j} = \frac{\mu_V \pi_{\tau+1}\bar{p}^{N}}{p^2}.
\end{equation}

\end{document}